\documentclass[]{fairmeta}

\usepackage[utf8]{inputenc}
\usepackage[T1]{fontenc}
\usepackage{url}
\usepackage{amsfonts}
\usepackage{amsmath}
\usepackage{amssymb}
\usepackage{amsthm}
\newtheorem{remark}{Remark}
\newtheorem{corollary}[remark]{Corollary}
\usepackage{nicefrac}
\usepackage{enumitem}
\usepackage{tikz}
\usepackage{pifont}
\usepackage{float}
\newcommand{\cmark}{\ding{51}}
\newcommand{\xmark}{\ding{55}}

\definecolor{prismblue}{HTML}{2B5EA7}
\definecolor{prismbg}{HTML}{F5F7FB}
\newtcolorbox{principlebox}[2][]{%
  enhanced,
  colback=prismbg,
  colframe=prismbg,
  boxrule=0pt,
  borderline west={3pt}{0pt}{prismblue},
  arc=0pt, outer arc=0pt,
  left=10pt, right=10pt, top=7pt, bottom=7pt,
  before skip=8pt, after skip=4pt,
  before upper={\noindent{\large\bfseries\sffamily\color{prismblue}#2}\par\smallskip},
  #1
}

\definecolor{plotblue}{HTML}{4C72B0}
\definecolor{plotred}{HTML}{C44E52}
\definecolor{plotgreen}{HTML}{55A868}
\definecolor{plotorange}{HTML}{DD8452}
\definecolor{plotpurple}{HTML}{8172B3}

\title{Computer Use at the Edge of the \\ Statistical Precipice}

\author[1]{Pierluca D'Oro}
\author[1]{Sneha Silwal}
\author[1]{William Wong}
\author[1]{Yuxuan Sun}
\author[1]{Fanyi Xiao}
\author[1]{Manchen Wang}
\author[1]{Eric Gan}
\author[1]{Allen Bolourchi}
\author[1]{Joseph Tighe}

\affiliation[1]{Meta Superintelligence Labs}

\abstract{Evaluating Computer Use Agents (CUAs) on interactive environments is fraught with methodological pitfalls that the field has yet to systematically address. We show that a 1MB replay script that blindly executes a recorded action sequence without ever observing the screen outperforms frontier models on prominent static benchmarks, and prove that its expected success rate is exactly equal to the source agent's pass@$k$ in deterministic environments. We trace this and other failures to two root causes: non-principled environment design (static, unsandboxed, or unreliably verified environments) and non-principled evaluation methodology (na\"ive aggregation and misuse of pass@$k$ for stateful UI interactions). To address the first, we propose PRISM, five design principles for CUA environments (privileged verification, realistic environments, integrity-checked configurations, sandboxed execution, and multifactorial variability) and instantiate them in \textsc{DigiWorld}, a benchmark of 15 realistic sandboxed mobile applications able to evaluate agents in over 3.2 million verified unique configurations. To address the second, we develop an aggregation framework pairing Wilson score intervals with hierarchical bootstrap, producing confidence intervals that correctly account for the nested structure of CUA benchmarks, as we empirically demonstrate. All together, we show that principled environment design and rigorous evaluation methodology are not optional refinements but prerequisites for meaningful CUA research.}

\date{May 2026}
\correspondence{\email{pierluca@meta.com, tighej@meta.com}}

\begin{document}

\maketitle

\section{Introduction}
\label{sec:intro}

What does it take for a Computer Use Agent (CUA) benchmark to produce reliable measurements? At minimum, the environment must be realistic, sandboxed, and varied enough to prevent memorization; its configurations must be verified as solvable; and task completion must be checked against internal application state rather than brittle proxies. Upon careful inspection, among existing major CUA benchmarks~\citep{xie2024osworld, rawles2024, drouin2024workarenacapablewebagents, koh2024visualwebarena, garg2025real, he2024webvoyager, sun2025, kong2025, ullrich2025, liu2018}, \emph{none} of them satisfies all of these design principles.

The consequences of these design gaps are severe and interconnected. Environments that depend on live services, drift unpredictably, making results irreproducible. Toy or synthetic interfaces lack the ecological validity to predict deployment performance. Static environments, where each task begins from a single fixed state, allow agents to succeed via memorized action sequences rather than genuine visual reasoning. As configuration spaces grow through parameterization, the absence of automated integrity checking means that a substantial fraction of tasks may be impossible, incoherent, or trivially pre-solved, silently corrupting aggregate metrics. And reliance on LLM-as-judge or screenshot comparison for verification introduces biases and adversarial vulnerabilities~\citep{zheng2023judging}. Concurrent work confirms that these failures are not theoretical: systematic audits have shown that major benchmarks can be exploited to achieve near-perfect scores without solving any tasks~\citep{wang2026}, and that such exploitation is already widespread~\citep{stein2026detecting}.

Compounding the environmental problem, \emph{evaluation methodology} introduces its own failures. Reporting a single success rate, often without confidence intervals or with intervals computed as if each trial were an independent coin flip, produces rankings that are unreliable under replication. The problem is worsened by the widespread misuse of pass@$k$ \citep{chen2021evaluatinglarge}, a metric designed for \emph{stateless} code generation that is fundamentally inappropriate for \emph{stateful} UI interactions.

This paper addresses both the environmental and methodological problems:

\begin{enumerate}[leftmargin=*,itemsep=2pt]
\item \textbf{Diagnosing pitfalls of current benchmarks.} We empirically demonstrate that, on popular deterministic benchmarks, outrageously small agents (of about 1MB of size) that replay successful action sequences from a frontier model can outperform the corresponding source model. We formalize this observation by establishing that the expected success rate of such a replay agent is exactly equal to the source agent's pass@$k$ (Remark~\ref{rem:replay}), revealing that pass@$k$ on static benchmarks measures memorization capacity rather than agent capability.

\item \textbf{PRISM design principles and \textsc{DigiWorld}.} We propose five design principles for CUA environments (\emph{privileged} verification, \emph{realistic} environments, \emph{integrity-checked} configurations, \emph{sandboxed} execution, and \emph{multifactorial} variability; PRISM) that benchmarks must satisfy to produce meaningful evaluations, and we instantiate them in \textsc{DigiWorld}, a combinatorial benchmark featuring 15 sandboxed mobile applications across domains. With 387 scenarios and over 4,300 task instances, \textsc{DigiWorld} independently varies data profiles (${\sim}10$ per app), visual themes (${\sim}10$), and initial UI states (${\sim}10$), generating over 3.2 million verified unique configurations, which leads to a benchmark robust to replay and memorization.

\item \textbf{Hierarchical aggregation framework.}
Environment variability is necessary for robust evaluation but, paired with stochasticity from an agent's actions, requires a rigorous statistical framework for confidence interval computation.
Since existing na\"ive aggregation strategies fail to handle this complexity, we develop a rigorous evaluation methodology that uses Wilson score intervals at the rollout level and a hierarchical bootstrap at the suite level, producing confidence intervals that correctly account for the nested structure of interactive benchmarks. We validate this framework through simulations and use it to evaluate frontier models on \textsc{DigiWorld}, showing that it offers better statistical coverage than widespread aggregation techniques and that it allows to rigorously track performance variability along each environmental axis.
\end{enumerate}

Our title is an homage to \citet{agarwal2021deep}, which exposed analogous statistical failures in deep RL evaluation. Where that work showed that point estimates of mean scores across a handful of Atari games produced unreliable agent rankings, we show that the CUA evaluation ecosystem suffers from a \emph{compounded} version of the same issue, non-principled environments feeding non-principled evaluation methodology, and we propose concrete solutions to both.

\section{Common pitfalls of CUA environment and metric design}
\label{sec:diagnosis}

Existing CUA benchmarks exhibit several recurring environmental weaknesses: unreliable verification via LLM-as-judge~\citep{zheng2023judging, wang2026}, synthetic interfaces that lack ecological validity~\citep{ullrich2025, liu2018}, absence of automated integrity checking for generated configurations, dependence on live environments that drift over time~\citep{he2024webvoyager, drouin2024workarenacapablewebagents}, and static single-configuration setups vulnerable to memorization~\citep{zhang2018studyoverfittingdeepreinforcement, pmlr-v119-cobbe20a}. Compounding these, evaluation \emph{methodology} introduces its own failures: na\"ive statistical aggregation~\citep{agarwal2021deep} and misuse of pass@$k$ produce rankings that cannot be replicated (Section~\ref{sec:fragility}). We now demonstrate the two most consequential failures empirically: the vulnerability of static environments to memorization, and the statistical fragility of standard reporting practices.

\subsection{Outrageously small agents can succeed on static benchmarks}
\label{sec:overfitting}

\textbf{Setup.}~~When a deterministic benchmark uses a single fixed initial state per task, blindly replaying a previously successful action sequence should succeed at a high rate.
We construct a \emph{replay agent}: for each task, we record the action sequence from a successful trajectory produced by a frontier model, then replay that exact sequence on subsequent evaluations without observing the state of the screen. The replay agent performs no visual reasoning; it simply executes the same clicks, taps, and keystrokes in the same order. We evaluate the replay agent on two prominent CUA benchmarks, \textsc{OSWorld} and \textsc{MobileWorld}, and on \textsc{DigiWorld} (our own benchmark described in Section~\ref{sec:digiworld}).

\begin{table}[t]
  \caption{Success rate of source frontier model and replay agent. Blindly replaying a recorded trajectory achieves high success rates on static benchmarks but collapses on \textsc{DigiWorld}.}
  \label{tab:overfitting}
  \centering
  \begin{tabular}{lccc}
    \toprule
    & \textsc{OSWorld} & \textsc{MobileWorld} & \textsc{DigiWorld}  \\
    \midrule
    Frontier model & 70.6\% & 32.7\% & 45.7\% \\
    $\sim$ 1MB-script replay agent  & 71.1\% & 41.5\% & 6.90\% \\
    \bottomrule
  \end{tabular}
\end{table}

\textbf{Results.}~~ Table~\ref{tab:overfitting} shows that the replay agent achieves higher success rate than the corresponding frontier model on static benchmarks, despite never observing the screen. Because the environment resets to the same initial state every time, the same action sequence produces the same outcome (apart from small mismatches due to benchmark flakiness). On \textsc{DigiWorld}, where each evaluation presents a novel combination of data profile, theme, and UI state, the replay agent collapses to near zero: the recorded actions no longer correspond to the correct interface elements. This confirms that static benchmarks are largely measuring whether the correct action sequence is \emph{known}, not whether the agent can \emph{reason} about what it sees.
See Appendix~\ref{app:replay_details} for evaluation details.

\textbf{Formal grounding.}~~ The replay agent's success is not accidental but structurally guaranteed. Consider a deterministic benchmark with $N$ tasks and a stochastic source agent $\pi$ with success probability $p_i$ on task $i$. Let $\pi_R$ be the replay policy constructed from $k$ independent rollouts of $\pi$: for each task, $\pi_R$ stores one successful trajectory (if any) and replays it verbatim. We write $\mathrm{SR}(\cdot)$ for the success rate averaged over tasks.

\begin{remark}[Replay Equivalence]
\label{rem:replay}
On a deterministic benchmark with fixed initial states, a replay policy that stores and replays one successful trajectory per task (if any exists among $k$ rollouts) achieves an expected success rate equal to the source agent's \emph{pass@}$k$:
\[
\mathbb{E}[\mathrm{SR}(\pi_R)] = \mathrm{pass@}k(\pi).
\]
Furthermore, $\mathrm{SR}(\pi) \leq \mathrm{pass@}k(\pi)$ for all $k \geq 1$, with strict inequality whenever $k > 1$ and $p_i \in (0,1)$ for some task~$i$.
\end{remark}

\begin{corollary}[Brute-Force Solvability of Static Benchmarks]
\label{cor:bruteforce}
For any agent with $p_i > 0$ on all $N$ tasks, $\lim_{k \to \infty} \mathbb{E}[\mathrm{SR}(\pi_R)] = 1$. Every static benchmark can be solved to arbitrary precision by running a stochastic agent for sufficiently many rollouts and memorizing successful trajectories.
\end{corollary}

\noindent Proofs are in Appendix~\ref{app:replay_proof}. Remark~\ref{rem:replay} reveals that the replay agent is the physical embodiment of pass@$k$: a constructive, memoryless algorithm that converts any nonzero success probability into near-perfect benchmark performance. This establishes a formal link between the environmental vulnerability demonstrated above and the methodological problems discussed next: \emph{reporting pass@$k$ on a static benchmark is equivalent to reporting the success rate of a blind replay policy}.

\textbf{Implications.}~~ Although our demonstration is purposely extreme, it points to the fact that any benchmark score obtained on a static environment is confounded with the degree to which the agent (or its training data) may have been inadvertently exposed to the specific evaluation trajectories. Improvements on such benchmarks may reflect better memorization rather than better visual grounding or capabilities. This does not mean that these benchmarks are without value, as they contain carefully designed tasks that test real capabilities, but that their \emph{environments} must be redesigned to prevent this type of exploitation. Multifactorial variability (Principle~M, Section~\ref{sec:principles}) breaks the replay equivalence by ensuring the initial state differs across evaluations along multiple independent axes.

\subsection{Point estimates and pass@$k$ as fragile CUA evals foundation}
\label{sec:fragility}

The environmental vulnerability above is compounded by methodological failures in how results are reported and aggregated. We identify three specific failures:

\textbf{Failure 1: Treating tasks as i.i.d.\ trials.}~~ A benchmark with $N$ tasks, each attempted once, is \emph{not} equivalent to $N$ coin flips. Tasks within the same application share UI flows, navigation patterns, and interaction primitives, and agent success is correlated across tasks that exercise similar capabilities~\citep{agarwal2021deep}. CIs that ignore this correlation dramatically understate uncertainty, leading to under-coverage rates exceeding 50\% on realistic benchmark structures (Section~\ref{sec:simulations}).

\textbf{Failure 2: Misuse of pass@$k$.}~~
The pass@$k$ estimator \citep{chen2021evaluatinglarge} was designed for code generation, where an evaluator can sample $k$ candidate programs, run each in isolation, and report whether \emph{any} candidate passes the tests. This protocol is meaningful because unsuccessful candidates do not alter the problem instance, and the evaluator can select the successful artifact after the fact. CUA trajectories do not have this structure: a rollout is not a candidate artifact; it is an execution that changes application state.
The operational quantity of interest is single-attempt success, not whether success appears at least once among $k$ attempts.
Moreover, as Remark~\ref{rem:replay} makes concrete: on a static benchmark, pass@$k$ is formally identical to the expected success rate of a replay policy that stores one successful trajectory and later replays it blindly. Reporting pass@$k$ in this setting measures the ability to find and reuse a successful trajectory, not reliable deployment capability.

\textbf{Failure 3: Ignoring the nested data structure.}~~ CUA benchmarks have an inherently hierarchical structure: suites contain apps, apps contain scenarios, scenarios may be evaluated across configurations, and each configuration is attempted across rollouts. Collapsing this hierarchy into a single fraction discards information about \emph{where} variability arises and produces confidence intervals that are either too narrow or too wide. We validate these concerns in Section~\ref{sec:simulations} and Section~\ref{sec:splithalf}.

\section{Principles for Computer Use Environment Design}
\label{sec:principles}

The diagnosis in Section~\ref{sec:diagnosis} demonstrates that current benchmarks fail on two fronts: their environments are designed without principled safeguards, and their evaluation methodology does not withstand scrutiny. We begin the constructive part of the paper by proposing five design principles that a CUA benchmark environment must satisfy. Together, these principles form the \textbf{PRISM} framework: \emph{privileged} verification, \emph{realistic} environments, \emph{integrity-checked} configurations, \emph{sandboxed} execution, and \emph{multifactorial} variability.

\begin{tcolorbox}[
  enhanced,
  colback=prismbg,
  colframe=prismblue,
  boxrule=0.8pt,
  arc=2pt,
  left=10pt, right=10pt, top=8pt, bottom=8pt,
  title={\sffamily\bfseries\color{white} PRISM Design Principles},
  coltitle=white,
  colbacktitle=prismblue,
  fonttitle=\sffamily\bfseries,
]
\begin{description}[leftmargin=0pt, itemsep=4pt, font=\sffamily\bfseries\color{prismblue}\small]
\item[P\;---\;Privileged verification.] Task completion must be verified as much as possible against the environment's internal ground-truth state.
\item[R\;---\;Realistic environments.] Benchmark applications must faithfully reproduce the visual appearance, interaction patterns, and functional complexity of real-world applications.
\item[I\;---\;Integrity-checked configurations.] Every environment configuration must be automatically verified, as feasible (preconditions satisfied), coherent (task-specific data consistent with existing state), and non-trivial (task not already completed in the initial state).
\item[S\;---\;Sandboxed execution.] The environment must be fully self-contained, with no external dependencies, and must support deterministic reset to any specified initial state.
\item[M\;---\;Multifactorial variability.] The environment must vary independently along multiple axes relevant to the capabilities being measured, such as data content, visual appearance, and initial application state.
\end{description}
\end{tcolorbox}

\noindent Each principle addresses a specific failure mode identified in Section~\ref{sec:diagnosis}. \emph{Privileged verification} eliminates the biases and adversarial vulnerabilities of LLM-as-judge and screenshot-based verification~\citep{zheng2023judging, wang2026}; privileged access to internal state makes verification deterministic and tamper-proof. \emph{Realistic environments} ensure ecological validity, preventing the gap between toy interfaces and real-world deployment. \emph{Integrity checking} is a structural necessity of any combinatorial benchmark: without it, a substantial fraction of generated configurations may be impossible or trivially pre-solved, silently corrupting aggregate metrics. \emph{Sandboxed execution} ensures reproducibility: benchmarks that depend on live services produce results that drift with UI updates, A/B tests, data changes, and network or software interruptions. Finally, \emph{multifactorial variability} prevents agents from succeeding by replaying memorized action sequences (Section~\ref{sec:overfitting}) and enables fine-grained diagnostics by decomposing performance along each axis of variation.

\subsection{Evaluating existing benchmarks against PRISM}
\label{sec:prism_eval}

Table~\ref{tab:prism_eval} evaluates ten major visual CUA benchmarks against the PRISM criteria (Appendix~\ref{app:prism_details} provides detailed per-benchmark analysis with rubric definitions). No existing benchmark satisfies all five principles. The most common gap is Integrity (I): while several benchmarks include initialization logic that provides partial safeguards, none really provides an automated offline pipeline that verifies every configuration is feasible, coherent, and non-trivial. This gap is especially dangerous in live or evolving environments where data drift can silently make tasks impossible. The Multifactorial (M) axis is also broadly absent: most benchmarks offer at most single-axis parameterization without independently varying visual appearance or data context.

\begin{table}[t]
  \caption{Evaluation of visual CUA benchmarks against the PRISM principles. \cmark{} = fully satisfies, {\raise.17ex\hbox{$\scriptstyle\sim$}} = partially satisfies, \xmark{} = does not satisfy. No existing benchmark satisfies all five principles; \textsc{DigiWorld} is designed to be the first.}
  \label{tab:prism_eval}
  \centering
  \resizebox{\textwidth}{!}{%
  \begin{tabular}{llccccc}
    \toprule
    Benchmark & Platform & \textbf{P}rivileged & \textbf{R}ealistic & \textbf{I}ntegrity & \textbf{S}andboxed & \textbf{M}ultifactorial \\
    \midrule
    AndroidWorld \citep{rawles2024} & Mobile & \cmark & \cmark & {\raise.17ex\hbox{$\scriptstyle\sim$}} & \cmark & {\raise.17ex\hbox{$\scriptstyle\sim$}} \\
    OSWorld \citep{xie2024osworld} & Desktop & \cmark & \cmark & {\raise.17ex\hbox{$\scriptstyle\sim$}} & \cmark & \xmark \\
    VisualWebArena \citep{koh2024visualwebarena} & Web & \cmark & \cmark & \xmark & \cmark & {\raise.17ex\hbox{$\scriptstyle\sim$}} \\
    WorkArena \citep{drouin2024workarenacapablewebagents} & Enterprise & \cmark & \cmark & \xmark & \xmark & {\raise.17ex\hbox{$\scriptstyle\sim$}} \\
    REAL-bench \citep{garg2025real} & Web & {\raise.17ex\hbox{$\scriptstyle\sim$}} & \cmark & \xmark & \cmark & \xmark \\
    DigiData-Bench \citep{sun2025} & Mobile & \xmark & \cmark & \xmark & \xmark & \xmark \\
    WebVoyager \citep{he2024webvoyager} & Web & \xmark & \cmark & \xmark & \xmark & \xmark \\
    MobileWorld \citep{kong2025} & Mobile & \cmark & \cmark & {\raise.17ex\hbox{$\scriptstyle\sim$}} & \cmark & \xmark \\
    OpenApps \citep{ullrich2025} & Mobile/Web & \cmark & {\raise.17ex\hbox{$\scriptstyle\sim$}} & \xmark & \cmark & \cmark \\
    MiniWoB++ \citep{liu2018} & Web & \cmark & \xmark & \xmark & \cmark & {\raise.17ex\hbox{$\scriptstyle\sim$}} \\
    \midrule
    \textsc{DigiWorld} (ours) & Mobile & \cmark & \cmark & \cmark & \cmark & \cmark \\
    \bottomrule
  \end{tabular}%
  }
\end{table}

\section{From point estimates to principled aggregation}
\label{sec:metrics}

\subsection{Sources of variability in CUA evaluation}
\label{sec:hierarchy}

CUA evaluation is characterized by two fundamentally different kinds of variability. \textbf{Action stochasticity} arises from the agent itself: given the exact same screen and task, different samples produce different action sequences. \textbf{Environmental variability} arises from the benchmark: different data, visual styles, or starting screens present the agent with different problems. Distinguishing these is critical because they answer different questions: ``how reliably does the agent solve \emph{this} problem?''\ versus ``how robust is the agent \emph{across} problems?''

Moreover, environments satisfying the PRISM principles feature stochasticity at multiple levels, which should be handled carefully by the statistical methodology used for agent evaluation. In this section, we develop an aggregation framework that properly handles performance estimation in such stochastic environments, and that fits the natural hierarchy of CUA benchmarks. The framework is general: it applies to any PRISM-compliant benchmark, not only to the one presented in this paper.

Any PRISM-compliant benchmark produces data with a nested hierarchy. At the coarsest level, a benchmark contains a suite of \textbf{apps} ($A$), each of which contains \textbf{scenarios} ($S_a$), i.e., task templates. Each scenario can be evaluated under multiple \textbf{configurations} ($C_{a,s}$), defined by the Cartesian product of whichever environmental axes the benchmark varies (e.g., data content, visual appearance, initial state). Finally, each configuration is attempted across multiple \textbf{rollouts} ($R$), which differ only in the agent's stochastic sampling of actions. Each rollout produces a binary outcome $Y_{a,s,c,r} \in \{0, 1\}$.

A sensible choice for environmental axes for CUA are instance parameters, data profiles, themes, and UI states (see Appendix~\ref{app:variability_axes}).
When all axes are disabled, there is exactly one configuration per scenario and the data reduce to a matrix of shape $A \times \bar{S} \times R$, analogous to the games $\times$ seeds structure in RL evaluation \citep{agarwal2021deep}. As axes are enabled, the configuration space grows combinatorially and the data form a tree of depth four. An effective aggregation framework should handle both extremes and everything in between. We address this bottom-up: first the leaf level (per-configuration intervals, Section~\ref{sec:wilson}), then the full tree (suite-level aggregation, Section~\ref{sec:bootstrap}).

\subsection{Wilson score intervals for binary outcomes}
\label{sec:wilson}

For a given configuration $(a, s, c)$ with $R$ rollouts, the empirical success rate is $\hat{p}_{a,s,c} = \frac{1}{R}\sum_{r=1}^{R} Y_{a,s,c,r}$. A common choice for computing a confidence interval for this success rate is the standard Wald interval $\hat{p} \pm z\sqrt{\hat{p}(1-\hat{p})/R}$. However, the Wald interval has poor coverage at extreme success rates and small $R$ \citep{brown2001}: it collapses to a zero-width degenerate interval when $\hat{p} = 0$ or $\hat{p} = 1$. To avoid this, we propose to adopt the use of the Wilson score interval~\citep{wilson1927}, which is a simple algebraic transformation of the Wald interval that provides near-nominal coverage even for small $R$ and extreme $\hat{p}$:
\begin{equation}
\hat{p}_W = \frac{\hat{p} + \frac{z^2}{2R}}{1 + \frac{z^2}{R}}, \qquad
W = \frac{z}{1 + \frac{z^2}{R}} \sqrt{\frac{\hat{p}(1-\hat{p})}{R} + \frac{z^2}{4R^2}}
\label{eq:wilson}
\end{equation}
where $z = z_{\alpha/2}$ is the standard normal quantile. The Wilson interval $[\hat{p}_W - W, \hat{p}_W + W]$ provides near-nominal coverage even for $R$ as small as 3 and $\hat{p}$ near 0 or 1.

\subsection{Hierarchical bootstrap}
\label{sec:bootstrap}

\textbf{Suite-level aggregation.}~~ Collapsing all outcomes into a single fraction treats every rollout as an independent coin flip, ignoring the correlation structure diagnosed in Section~\ref{sec:fragility}. Following \citet{agarwal2021deep}, who showed that na\"ive averaging across domains (e.g. Atari games) produces unreliable RL agent rankings, we propose to aggregate in two stages. First, by computing a per-app success rate $\bar{p}_a$ by averaging over all scenarios, configurations, and rollouts within app $a$. Then, by summarizing the suite as $\hat{\theta} = \frac{1}{A}\sum_{a=1}^{A}\bar{p}_a$\label{eq:suite}, the mean across apps. Because apps are a curated population rather than a random sample, we treat them as fixed and never resample them; this also ensures aggregation fairness, since each app contributes equally regardless of how many scenarios it contains.

\textbf{Confidence intervals via hierarchical bootstrap.}~~ Standard bootstrap CIs that resample only at the rollout level dramatically understate uncertainty because they miss the variance contributed by scenario difficulty and configuration effects (Section~\ref{sec:fragility}, Failure~3). We therefore propose a hierarchical bootstrap that resamples at each active level \emph{within} each app: (1)~resample $S_a$ scenarios with replacement, (2)~within each scenario, resample along each enabled environmental axis independently, and (3)~within each configuration, resample $R$ rollouts. In the absence of varied environment configurations, steps (1) and (3) remain and step (2) becomes a no-op, reducing to a stratified bootstrap analogous to \texttt{rliable} with \texttt{task\_bootstrap=True} \citep{agarwal2021deep}, but applied independently within each app. For each replicate $b$, we recompute $\hat{\theta}^{(b)}$ and construct the 95\% CI from the bootstrap quantiles. We provide a simulation study in Section~\ref{sec:simulations} that confirms all resampling levels are necessary (e.g., omitting scenario resampling produces under-coverage).

\section{\textsc{DigiWorld}: Following PRISM for Mobile Computer Use}
\label{sec:digiworld}

We now present \textsc{DigiWorld}, a benchmark designed to satisfy all five PRISM principles: it is \emph{realistic}, with each application faithfully reproducing real-world interaction patterns; it is \emph{sandboxed}, being fully self-contained with no external dependencies and featuring deterministic reset via internal application state modification; it is \emph{multifactorial}, supporting independent variation of data profiles, themes, UI states, and instance parameters; it is \emph{privileged}, allowing verification via parameterized SQL queries against internal application state; it is \emph{integrity-checked}, running automated feasibility checks and triviality filters to ensure every configuration is well-posed.

\textsc{DigiWorld} is made up of 15 hand-built, full-featured Android applications spanning 7 real-world domains. It comprises 387 scenarios (304 original single-app tasks plus 83 composed multi-step tasks) with over 4,300 task instances and a framework for defining and verifying tasks. Its hierarchical structure plugs directly into the evaluation methodology of Section~\ref{sec:metrics}.

\textbf{Fifteen apps across seven domains.}~~%
\label{sec:apps}%
The 15 applications span 7 real-world domains (Table~\ref{tab:apps}, Appendix~\ref{app:digiworld_stats}): communication (Email, Messaging), financial services (Payment, Banking), shopping (E-Commerce, QwikShop, Auction), entertainment (Music, Video), on-demand services (Eats, Ryde), travel (Flight Booking, Transit), and smart living (Smart Home, Parking). Each application is a fully-featured React Native Android application, hand-built for this benchmark, with a local database for the data state of the application, a JSON file representing its UI state, and a UI theming system.

\textbf{The scenario lifecycle.}~~%
\label{sec:scenarios}%
A \emph{scenario} defines a parameterized task template (e.g., ``Send \$\texttt{<amount>} to \texttt{<recipient>}'') together with its verification logic. Each scenario is instantiated into multiple \emph{instances} by binding concrete values to the parameters; each instance is compatible with a subset of data profiles, determined by automated feasibility checking (described below and detailed in Appendix~\ref{app:integrity}). A \emph{configuration} is the combination of an instance with a specific data profile, theme, and UI state. The lifecycle has three phases: (1)~\emph{reset} (inject the selected profile's database and instance-specific mockdata), (2)~\emph{execution} (the agent interacts through the standard Android GUI via screen-based actions (tap, type, scroll, with no programmatic API or tool access), and (3)~\emph{verification} (compare final application state to initial state). \textsc{DigiWorld} is agent-agnostic: any system that can send GUI events and receive screenshots can be evaluated. Appendix~\ref{app:scenario_lifecycle} provides full details.

\textbf{Integrity at scale.}~~%
\label{sec:integrity}%
\label{sec:context}%
\label{sec:feasibility}%
Multifactorial variation generates a vast configuration space (over 3.2 million configurations across the suite), making manual curation of valid combinations infeasible. Programmatic integrity checking is therefore a \emph{structural necessity} of any combinatorial benchmark. \textsc{DigiWorld} addresses this through three automated pipelines that run offline before any agent is evaluated: (1)~\emph{coherence checking}, which resolves template placeholders in instance-specific mockdata against each profile's actual database to ensure referential consistency; (2)~\emph{feasibility checking}, which evaluates parameter-dependent constraints (e.g., sufficient balance for a transfer) to produce per-instance lists of compatible profiles; and (3)~\emph{triviality filtering}, which excludes configurations where the task is already completed in the initial state. Together, these mechanisms enable the combinatorial design to scale without sacrificing integrity guarantees. Appendix~\ref{app:integrity} provides implementation details.

\section{Experiments}
\label{sec:experiments}

\subsection{Validating the bootstrap via simulation}
\label{sec:simulations}

Our evaluation methodology makes two complementary improvements: a better base-level estimator (Wilson score intervals) and a hierarchical aggregation that respects the nested data structure. \begin{figure}[t]
  \centering
  \includegraphics[width=\textwidth]{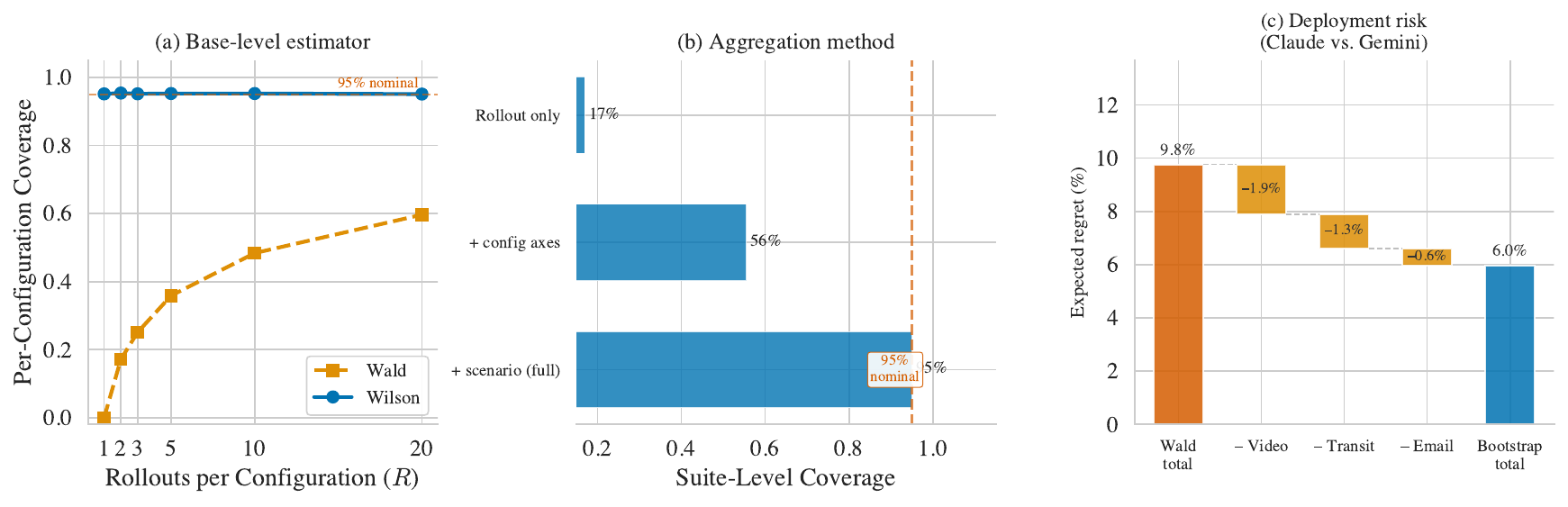}
  \caption{Validating the evaluation framework. \textbf{(a)}~Per-configuration coverage of Wald vs.\ Wilson intervals across rollout budgets $R$, with parameters calibrated from \textsc{DigiWorld} evaluations (Appendix~\ref{app:sim_calibration}). Wilson achieves near-nominal (95\%) coverage at all $R$; Wald under-covers severely, reaching only 25\% at $R{=}3$. \textbf{(b)}~Suite-level coverage under heterogeneous conditions. Each rung cumulatively adds one resampling level; only the full hierarchical bootstrap reaches 95\%. \textbf{(c)}~Per-app deployment regret (Claude vs.\ Gemini, Section~\ref{sec:splithalf}). Expected regret is the probability of deploying the wrong model times the true per-app gap. The three apps where only Wald CIs claim significance contribute 3.8\% of expected regret that the bootstrap prevents by correctly flagging them as uncertain.}
  \label{fig:simulation}
\end{figure}

\textbf{Base-level estimator.}~~ Figure~\ref{fig:simulation}a compares per-configuration coverage of Wald and Wilson intervals. True success probabilities are calibrated from \textsc{DigiWorld} evaluation data, which is heavily bimodal (${\sim}68\%$ of configurations at 0\% success, ${\sim}32\%$ at 100\%); we recover continuous true probabilities via a Jeffreys posterior (Appendix~\ref{app:sim_calibration}). At $R{=}3$ (the typical budget in CUA evaluation), the Wald interval covers the true value only 25\% of the time, because it produces degenerate zero-width CIs when $\hat{p} = 0$ or $\hat{p} = 1$. The Wilson interval maintains near-nominal coverage (95\%) across all $R$, including $R{=}1$, by producing non-degenerate intervals at the extremes.

\textbf{Hierarchical aggregation.}~~ Figure~\ref{fig:simulation}b presents a cumulative \emph{ladder} of bootstrap variants under heterogeneous conditions, with app-level success rates and noise parameters calibrated from \textsc{DigiWorld} evaluations (Appendix~\ref{app:sim_calibration}). A bootstrap that resamples only rollouts achieves just 17\% suite-level coverage. Adding configuration-axis resampling (capturing uncertainty about which environments the agent encounters) raises coverage to 56\%. Adding scenario resampling on top completes the full hierarchy at 95\% coverage, the only variant to reach the nominal level. Each level of the hierarchy captures genuine uncertainty that simpler methods miss.

\subsection{Experimental results}
\label{sec:splithalf}
\label{sec:decomposition}

\textbf{Evaluation protocol.}~~ For each of the 387 scenarios, we evaluate $R = 20$ rollouts per sampled configuration. The full configuration space exceeds 3.2 million; we sample configurations by drawing instances, data profiles, themes, and UI states from the available pools for each scenario. We use $B = 1{,}000$ bootstrap replicates for all confidence intervals.
We evaluate three frontier models: Claude Opus 4.6, GPT-5.4, and Gemini 3 Pro. All models were accessed via their public APIs in April 2026. Each model receives a task description, the current screenshot, and a set of available GUI actions (tap, type, scroll, press home/back, answer). No in-context demonstrations are provided.

\begin{figure}[t]
  \centering
  \includegraphics[width=\textwidth]{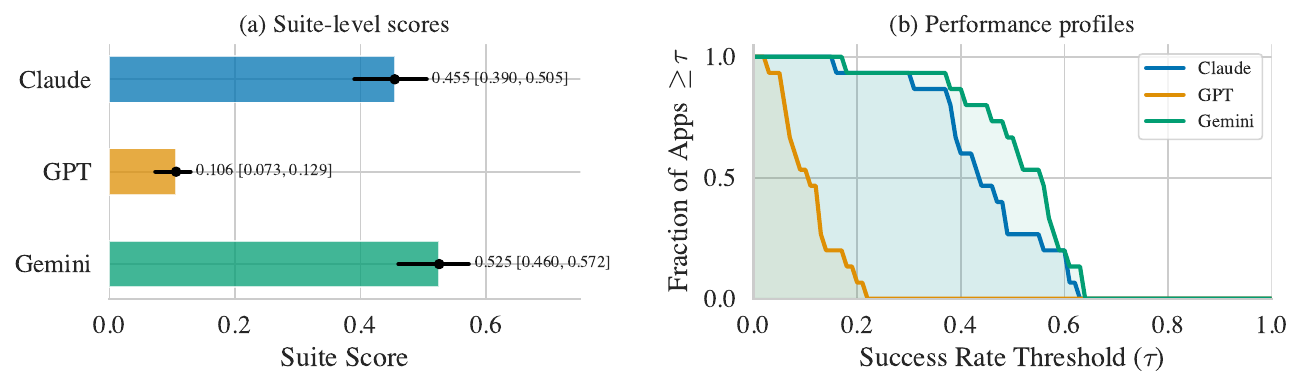}
  \caption{\textbf{(a)} Suite-level scores with 95\% hierarchical bootstrap CIs. Claude and Gemini overlap; GPT is clearly separated. \textbf{(b)} Performance profiles: each curve shows the fraction of apps on which a model achieves at least the success rate $\tau$.}
  \label{fig:results}
\end{figure}

\textbf{Aggregate results.}~~ Figure~\ref{fig:results}a shows suite-level scores with 95\% hierarchical bootstrap CIs. Claude and Gemini achieve similar scores with broadly overlapping CIs, indicating that the difference is not statistically meaningful; GPT struggles with mobile computer use as reported by previous work~\citep{kong2025}. Figure~\ref{fig:results}b shows performance profiles \citep{agarwal2021deep} across apps. Per-app results (Appendix~\ref{app:per_app}) show wide CIs even with thousands of rollouts, underscoring that app-level variance dominates the uncertainty budget. Results for additional model variants are in Appendix~\ref{app:additional_models}.

\textbf{Would na\"ive statistics get it wrong?}~~ Because Wald CIs underestimate uncertainty, they declare statistically significant differences for apps where the true gap is small and the ranking is fragile. To measure the downstream cost, we define \emph{expected regret} per app as the probability of selecting the inferior model (estimated via 500 split-half simulations at one configuration per scenario) multiplied by the true performance gap for that app. A practitioner who deploys whichever model a CI method declares ``significantly better'' incurs this regret on every app where the method is both overconfident and wrong. Figure~\ref{fig:simulation}c shows that the hierarchical bootstrap reduces total expected regret from 9.8\% to 6.0\% by correctly withholding significance on three apps whose rankings are too uncertain to act on. The 3.8\% of prevented regret is comparable to the gap between adjacent frontier models on current benchmarks, illustrating how miscalibrated CIs can silently erase real performance differences.

\textbf{Variability decomposition.}~~ The combinatorial structure of \textsc{DigiWorld} enables diagnostic analyses not possible on static benchmarks. We use matched-pair comparisons to isolate the effect of each environmental axis, summarizing instability via the mean absolute deviation (MAD) of performance shifts along axes (Appendix~\ref{app:variability}).
\begin{figure}[t]
  \centering
  \includegraphics[width=\textwidth]{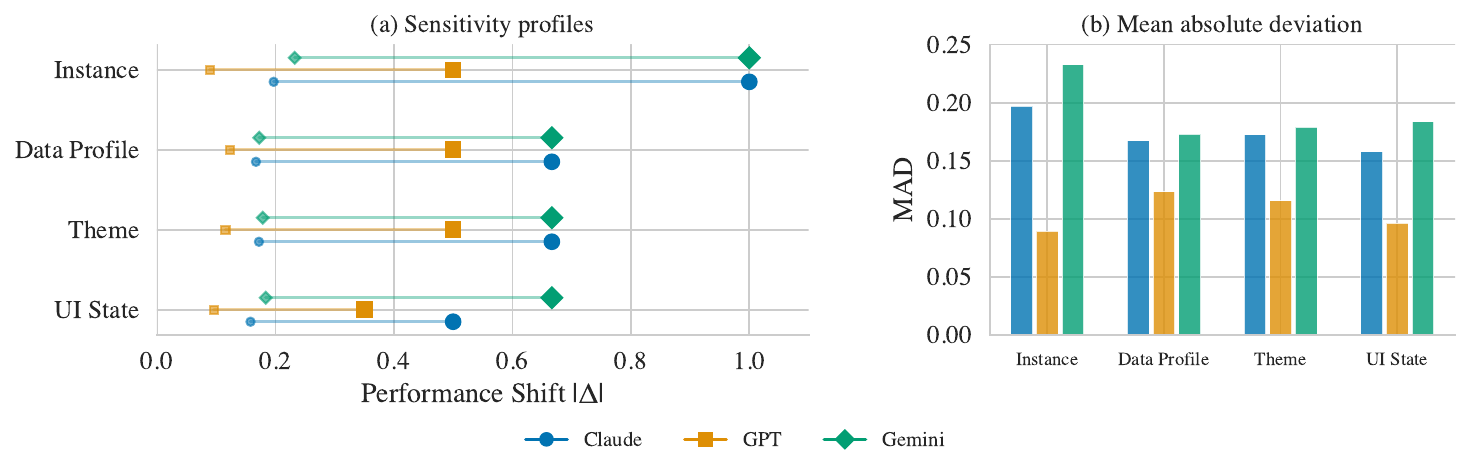}
  \caption{\textbf{(a)}~Sensitivity profiles: small marker = MAD (average instability), large marker = 90th percentile of $|\Delta|$ (worst case); longer lines indicate heavier tails. \textbf{(b)}~MAD per axis and model.}
  \label{fig:sensitivity}
\end{figure}
Figure~\ref{fig:sensitivity} shows that all axes produce MADs in the range 0.09--0.21, with no single axis dominating uniformly. This heterogeneity confirms that multifactorial variation is essential: a benchmark varying only one axis would miss substantial sources of instability.

\textbf{Replay resistance.}~~ Returning to the replay experiment of Section~\ref{sec:overfitting}, a replay agent only succeeds on \textsc{DigiWorld} (Table~\ref{tab:overfitting}) in simple navigation tasks (e.g., ``show contacts list,'' ``logout'') where the target UI element occupies a fixed position regardless of configuration. Instead, tasks requiring data-dependent reasoning (e.g., filtering, searching) fail almost completely. This confirms that the variability selectively defeats memorization while preserving solvability for genuine visual reasoning.

\section{Related work}
\label{sec:background}
\label{sec:discussion}

\textbf{CUA benchmarks.}~~ The past two years have seen a proliferation of CUA benchmarks. OSWorld \citep{xie2024osworld} provides tasks in a virtualized Ubuntu environment. AndroidWorld \citep{rawles2024} defines tasks on live Android applications. WebArena \citep{zhou2024webarena} and VisualWebArena \citep{koh2024visualwebarena} focus on web-based tasks. WorkArena \citep{drouin2024workarenacapablewebagents} targets enterprise software. REAL-bench \citep{garg2025real} constructs deterministic replicas of popular websites, WebVoyager \citep{he2024webvoyager} tests agents on live production websites, DigiData-Bench \citep{sun2025} evaluates mobile control agents on live apps, MobileWorld \citep{kong2025} benchmarks mobile agents with MCP-augmented tools, OpenApps \citep{ullrich2025} uses synthetic apps with explicit appearance variation, and MiniWoB++ \citep{liu2018} provides simplified web interaction tasks. While these benchmarks vary in scope and domain, they share common structural limitations; Table~\ref{tab:prism_eval} evaluates each against the PRISM criteria defined in Section~\ref{sec:principles}.

\textbf{Overfitting, exploitation, and statistical methodology.}~~ The risk of agents exploiting fixed environments is well-documented in games and robotics~\citep{zhang2018studyoverfittingdeepreinforcement, pmlr-v119-cobbe20a,tobin2017}, and environment unreliability has been observed for CUA~\citep{wang2026, stein2026detecting}. On the statistical side, \citet{agarwal2021deep} showed that point estimates without robust aggregation produce unreliable agent rankings, proposing a framework with stratified bootstrap CIs. Our work extends this line of work to CUA, where the data hierarchy is deeper and binary outcomes require different estimators. Specifically, we use Wilson score intervals~\citep{wilson1927, brown2001, capuano2026} and a hierarchical bootstrap to correctly handle the nested structure.

\section{Limitations, broader impact, and conclusion}
\label{sec:conclusion}

\textbf{Limitations.}~~ PRISM operates at the level of environment and benchmark design; it does not address lower-level concerns such as harness design, task-specific verifier engineering, or the practical challenges of sandboxing complex environments. These are essential but can only be addressed reasonably atop principled foundations. \textsc{DigiWorld} currently targets only mobile (Android) applications; extending PRISM to desktop and web environments requires significant engineering effort, as each platform demands its own sandboxing strategy, verification mechanisms, and application framework.

\textbf{Broader impact.}~~ We hope our work can redirect effort from benchmark-specific optimizations that may not transfer to deployment toward genuine understanding of agent capabilities. The variability decomposition framework provides a template for understanding \emph{why} agents fail, not just \emph{how often}.

We have shown that the CUA evaluation ecosystem suffers from two compounding pathologies: non-principled environment designs and evaluation methodologies. We address both through the PRISM design principles instantiated in \textsc{DigiWorld}, and through a hierarchical aggregation framework that produces confidence intervals by correctly modeling the nested data structure. The PRISM framework is not specific to \textsc{DigiWorld}: any interactive benchmark that satisfies these five principles can use the same hierarchical aggregation methodology. Together, these contributions provide the methodological foundation for CUA evaluation that is reliable, diagnostic, and resistant to gaming.


\bibliographystyle{assets/plainnat}
\bibliography{references}

@inproceedings{agarwal2021deep,
  title={Deep Reinforcement Learning at the Edge of the Statistical Precipice},
  author={Agarwal, Rishabh and Schwarzer, Max and Castro, Pablo Samuel and Courville, Aaron C. and Bellemare, Marc},
  booktitle={Advances in Neural Information Processing Systems},
  volume={34},
  pages={29304--29320},
  year={2021},
  url={https://proceedings.neurips.cc/paper_files/paper/2021/file/f514cec81cb148559cf475e7426eed5e-Paper.pdf}
}

@article{brown2001,
  author    = {Brown, Lawrence D. and Cai, T. Tony and DasGupta, Anirban},
  title     = {Interval Estimation for a Binomial Proportion},
  journal   = {Statistical Science},
  volume    = {16},
  number    = {2},
  pages     = {101--133},
  year      = {2001},
  doi       = {10.1214/ss/1009213286}
}

@misc{capuano2026,
  title = {A primer on measuring the uncertainty of success rates},
  author = {Capuano, Francesco},
  year = {2026},
  month = {March},
  howpublished = {Blog post},
  url = {https://fracapuano.github.io/blog/success-rates}
}

@misc{chen2021evaluatinglarge,
      title={Evaluating Large Language Models Trained on Code},
      author={Mark Chen and Jerry Tworek and Heewoo Jun and Qiming Yuan and Henrique Ponde de Oliveira Pinto and Jared Kaplan and Harri Edwards and Yuri Burda and Nicholas Joseph and Greg Brockman and Alex Ray and Raul Puri and Gretchen Krueger and Michael Petrov and Heidy Khlaaf and Girish Sastry and Pamela Mishkin and Brooke Chan and Scott Gray and Nick Ryder and Mikhail Pavlov and Alethea Power and Lukasz Kaiser and Mohammad Bavarian and Clemens Winter and Philippe Tillet and Felipe Petroski Such and Dave Cummings and Matthias Plappert and Fotios Chantzis and Elizabeth Barnes and Ariel Herbert-Voss and William Hebgen Guss and Alex Nichol and Alex Paino and Nikolas Tezak and Jie Tang and Igor Babuschkin and Suchir Balaji and Shantanu Jain and William Saunders and Christopher Hesse and Andrew N. Carr and Jan Leike and Josh Achiam and Vedant Misra and Evan Morikawa and Alec Radford and Matthew Knight and Miles Brundage and Mira Murati and Katie Mayer and Peter Welinder and Bob McGrew and Dario Amodei and Sam McCandlish and Ilya Sutskever and Wojciech Zaremba},
      year={2021},
      eprint={2107.03374},
      archivePrefix={arXiv},
      primaryClass={cs.LG}
}

@InProceedings{pmlr-v119-cobbe20a,
  title = {Leveraging Procedural Generation to Benchmark Reinforcement Learning},
  author = {Cobbe, Karl and Hesse, Chris and Hilton, Jacob and Schulman, John},
  booktitle = {Proceedings of the 37th International Conference on Machine Learning},
  pages = {2048--2056},
  year = {2020},
  editor = {III, Hal Daumé and Singh, Aarti},
  volume = {119},
  series = {Proceedings of Machine Learning Research},
  month = {13--18 Jul},
  publisher = {PMLR},
  url = {https://proceedings.mlr.press/v119/cobbe20a.html}
}

@misc{drouin2024workarenacapablewebagents,
      title={WorkArena: How Capable Are Web Agents at Solving Common Knowledge Work Tasks?},
      author={Alexandre Drouin and Maxime Gasse and Massimo Caccia and Issam H. Laradji and Manuel Del Verme and Tom Marty and Léo Boisvert and Megh Thakkar and Quentin Cappart and David Vazquez and Nicolas Chapados and Alexandre Lacoste},
      year={2024},
      eprint={2403.07718},
      archivePrefix={arXiv},
      primaryClass={cs.LG},
      url={https://arxiv.org/abs/2403.07718},
}

@article{garg2025real,
  author  = {Divyansh Garg and Shaun VanWeelden and Diego Caples and Andis Draguns and Nikil Ravi and Pranav Putta and Naman Garg and Tomas Abraham and Michael Lara and Federico Lopez and James Liu and Atharva Gundawar and Prannay Hebbar and Youngchul Joo and Jindong Gu and Charles London and Christian Schroeder de Witt and Sumeet Motwani},
  title   = {REAL: Benchmarking Autonomous Agents on Deterministic Simulations of Real Websites},
  journal = {arXiv preprint arXiv:2504.11543},
  year    = {2025},
}

@misc{he2024webvoyager,
      title={WebVoyager: Building an End-to-End Web Agent with Large Multimodal Models},
      author={Hongliang He and Wenlin Yao and Kaixin Ma and Wenhao Yu and Yong Dai and Hongming Zhang and Zhenzhong Lan and Dong Yu},
      year={2024},
      eprint={2401.13919},
      archivePrefix={arXiv},
      primaryClass={cs.CL}
}

@inproceedings{koh2024visualwebarena,
    title = "{V}isual{W}eb{A}rena: Evaluating Multimodal Agents on Realistic Visual Web Tasks",
    author = "Koh, Jing Yu and Lo, Robert and Jang, Lawrence and Duvvur, Vikram and Lim, Ming Chong and Huang, Po-Yu and Neubig, Graham and Zhou, Shuyan and Salakhutdinov, Ruslan and Fried, Daniel",
    booktitle = "Proceedings of the 62nd Annual Meeting of the Association for Computational Linguistics (Volume 1: Long Papers)",
    month = "aug",
    year = "2024",
    address = "Bangkok, Thailand",
    publisher = "Association for Computational Linguistics",
    url = "https://aclanthology.org/2024.acl-long.50",
    doi = "10.18653/v1/2024.acl-long.50",
    pages = "881--905"
}

@misc{kong2025,
      title={MobileWorld: Benchmarking Autonomous Mobile Agents in Agent-User Interactive and MCP-Augmented Environments},
      author={Quyu Kong and Xu Zhang and Zhenyu Yang and Nolan Gao and Chen Liu and Panrong Tong and Chenglin Cai and Hanzhang Zhou and Jianan Zhang and Liangyu Chen and Zhidan Liu and Steven Hoi and Yue Wang},
      year={2025},
      eprint={2512.19432},
      archivePrefix={arXiv},
      primaryClass={cs.CL}
}

@inproceedings{liu2018,
  author    = {Evan Zheran Liu and Kelvin Guu and Panupong Pasupat and Tianlin Shi and Percy Liang},
  title     = {Reinforcement Learning on Web Interfaces Using Workflow-Guided Exploration},
  booktitle = {International Conference on Learning Representations},
  year      = {2018},
  url       = {https://arxiv.org/abs/1802.08802},
}

@misc{rawles2024,
      title={AndroidWorld: A Dynamic Benchmarking Environment for Autonomous Agents},
      author={Christopher Rawles and Sarah Clinckemaillie and Yifan Chang and Jonathan Waltz and Gabrielle Lau and Marybeth Fair and Alice Li and William Bishop and Wei Li and Folawiyo Campbell-Ajala and Daniel Toyama and Robert Berry and Divya Tyamagundlu and Timothy Lillicrap and Oriana Riva},
      year={2024},
      eprint={2405.14573},
      archivePrefix={arXiv},
      primaryClass={cs.AI}
}

@article{stein2026detecting,
  title={Detecting Safety Violations Across Many Agent Traces},
  author={Stein, Adam and Brown, Davis and Hassani, Hamed and Naik, Mayur and Wong, Eric},
  journal={arXiv preprint arXiv:2604.11806},
  year={2026}
}

@misc{sun2025,
      title={DigiData: Training and Evaluating General-Purpose Mobile Control Agents},
      author={Yuxuan Sun and Manchen Wang and Shengyi Qian and William R. Wong and Eric Gan and Pierluca D'Oro and Alejandro Castillejo Munoz and Sneha Silwal and Pedro Matias and Nitin Kamra and Satwik Kottur and Nick Raines and Xuanyi Zhao and Joy Chen and Joseph Greer and Andrea Madotto and Allen Bolourchi and James Valori and Kevin Carlberg and Karl Ridgeway and Joseph Tighe},
      year={2025},
      eprint={2511.07413},
      archivePrefix={arXiv},
      primaryClass={cs.AI}
}

@inproceedings{tobin2017,
  author    = {Josh Tobin and Rachel Fong and Alex Ray and Jonas Schneider and Wojciech Zaremba and Pieter Abbeel},
  title     = {Domain Randomization for Transferring Deep Neural Networks from Simulation to the Real World},
  booktitle = {2017 IEEE/RSJ International Conference on Intelligent Robots and Systems (IROS)},
  pages     = {23--30},
  year      = {2017},
  doi       = {10.1109/IROS.2017.8202133},
  eprint    = {1703.06907},
  archivePrefix = {arXiv},
}

@article{ullrich2025,
  author    = {Karen Ullrich and Jingtong Su and Claudia Shi and Arjun Subramonian and Amir Bar and Ivan Evtimov and Nikolaos Tsilivis and Randall Balestriero and Julia Kempe and Mark Ibrahim},
  title     = {OpenApps: Simulating Environment Variations to Measure UI-Agent Reliability},
  journal   = {arXiv preprint arXiv:2511.20766},
  year      = {2025},
  eprint    = {2511.20766},
  archivePrefix = {arXiv},
  primaryClass = {cs.AI}
}

@misc{wang2026,
  author = {Hao Wang and Qiuyang Mang and Alvin Cheung and Koushik Sen and Dawn Song},
  title = {How We Broke Top AI Agent Benchmarks: And What Comes Next},
  howpublished = {Blog post},
  month = {April},
  year = {2026},
  url = {https://rdi.berkeley.edu/blog/trustworthy-benchmarks-cont/},
  note = {Accessed on 2026-04-16}
}

@article{wilson1927,
  title={Probable inference, the law of succession, and statistical inference},
  author={Wilson, Edwin B.},
  journal={Journal of the American Statistical Association},
  volume={22},
  number={158},
  pages={209--212},
  year={1927},
  publisher={Taylor & Francis},
  doi={10.1080/01621459.1927.10502953}
}

@inproceedings{xie2024osworld,
  title={OSWorld: Benchmarking Multimodal Agents for Open-Ended Tasks in Real Computer Environments},
  author={Xie, Tianbao and Zhang, Danyang and Chen, Jixuan and Li, Xiaochuan and Zhao, Siheng and Cao, Ruisheng and Hua, Toh Jing and Cheng, Zhoujun and Shin, Dongchan and Lei, Fangyu and Liu, Yitao and Xu, Yiheng and Zhou, Shuyan and Savarese, Silvio and Xiong, Caiming and Zhong, Victor and Yu, Tao},
  booktitle={Advances in Neural Information Processing Systems},
  year={2024},
  eprint={2404.07972},
  archivePrefix={arXiv},
  primaryClass={cs.AI}
}

@misc{zhang2018studyoverfittingdeepreinforcement,
      title={A Study on Overfitting in Deep Reinforcement Learning},
      author={Chiyuan Zhang and Oriol Vinyals and Remi Munos and Samy Bengio},
      year={2018},
      eprint={1804.06893},
      archivePrefix={arXiv},
      primaryClass={cs.LG},
      url={https://arxiv.org/abs/1804.06893},
}

@misc{zheng2023judging,
  title={Judging LLM-as-a-Judge with MT-Bench and Chatbot Arena},
  author={Lianmin Zheng and Wei-Lin Chiang and Ying Sheng and Siyuan Zhuang and Zhanghao Wu and Yonghao Zhuang and Zi Lin and Zhuohan Li and Dacheng Li and Eric P. Xing and Hao Zhang and Joseph E. Gonzalez and Ion Stoica},
  year={2023},
  eprint={2306.05685},
  archivePrefix={arXiv},
  primaryClass={cs.CL}
}

@inproceedings{zhou2024webarena,
  title={WebArena: A Realistic Web Environment for Building Autonomous Agents},
  author={Shuyan Zhou and Frank F. Xu and Hao Zhu and Xuhui Zhou and Robert Lo and Abishek Sridhar and Xianyi Cheng and Tianyue Ou and Yonatan Bisk and Daniel Fried and Uri Alon and Graham Neubig},
  booktitle={The Twelfth International Conference on Learning Representations},
  year={2024},
  url={https://openreview.net/forum?id=oKn9c6ytLx}
}


\newpage
\beginappendix

\section{Proof of Replay Equivalence}
\label{app:replay_proof}

\begin{proof}[Proof of Remark~\ref{rem:replay}]
\textbf{Part 1: $\mathbb{E}[\mathrm{SR}(\pi_R)] = \mathrm{pass@}k(\pi)$.}
Fix a task $i$. The $k$ rollouts are independent Bernoulli trials with success probability $p_i$. The replay policy $\pi_R$ succeeds on task $i$ if and only if at least one rollout succeeded (i.e., a successful trajectory exists to replay). Because the benchmark is deterministic and initializes task $i$ to the same state $s_i^0$ every time, replaying a previously successful action sequence is guaranteed to succeed again. Therefore:
\[
P(\pi_R \text{ succeeds on task } i) = 1 - (1 - p_i)^k.
\]
Averaging over all $N$ tasks gives $\mathbb{E}[\mathrm{SR}(\pi_R)] = \mathrm{pass@}k(\pi)$.

\medskip

\textbf{Part 2: $\mathrm{SR}(\pi) \leq \mathrm{pass@}k(\pi)$.}
It suffices to show $p_i \leq 1 - (1 - p_i)^k$ for each task. Let $q = 1 - p_i \in [0,1]$. We need $q^k \leq q$. For $q \in [0,1]$ and $k \geq 1$, this holds because $q^k = q \cdot q^{k-1} \leq q \cdot 1 = q$, with strict inequality when $q \in (0,1)$ and $k \geq 2$.
\end{proof}

\begin{proof}[Proof of Corollary~\ref{cor:bruteforce}]
For any $p_i > 0$, $(1 - p_i)^k \to 0$ as $k \to \infty$. Therefore $\mathbb{E}[\mathrm{SR}(\pi_R)] = \frac{1}{N}\sum_{i=1}^{N}[1 - (1-p_i)^k] \to 1$.
\end{proof}

\section{Replay agent evaluation details}
\label{app:replay_details}

For each benchmark in Table~\ref{tab:overfitting}, the replay agent stores successful trajectories from the best-performing frontier model on that benchmark: Claude Opus 4.6 for \textsc{OSWorld} and \textsc{DigiWorld}, and Gemini 3 Pro for \textsc{MobileWorld}. The ``frontier model'' row in the table reports that same model's success rate.

\subsection{MobileWorld}

From the 201 tasks in \textsc{MobileWorld}, we evaluate on a subset of 96 tasks. We exclude 105 tasks that our evaluation infrastructure cannot reliably support or that are instrumented inconsistently in \textsc{MobileWorld} itself. Specifically: (i)~40 tasks tagged \texttt{agent-mcp} require external MCP cloud APIs (Amap for maps/routing, Stockstar for financial data, arXiv/GitHub for information retrieval) for which we lack credentials; (ii)~41 Mastodon tasks and 16 Mattermost tasks rely on Docker-in-Docker backend services which were unstable on our virtual machines; and (iii)~8 mall shopping tasks fail deterministically across every run of every model with ``No callback data found,'' indicating a broken callback mechanism in the app. The remaining 96 tasks span 7 app categories (Calendar, Mail, Messages, Maps, Chrome, Files/Docreader, Taodian, Camera/Gallery, Clock, Settings, Contacts) and cover the full range of \textsc{MobileWorld}'s task complexity, including web-browsing tasks that require live internet access.

\section{Cost considerations}
\label{app:cost}

The combinatorial evaluation protocol we propose is substantially more expensive than evaluating on a static benchmark: over $8{,}000$ configurations per scenario vs.\ 1. We argue that this cost is the price of \emph{meaningful} evaluation. A cheap evaluation that does not measure what it claims to measure is not a bargain. Repeating rollouts on a single fixed configuration yields a tight confidence interval on \emph{that configuration's} success rate, but this is not the quantity of interest: what matters is the agent's expected performance across the full space of realistic conditions. The combinatorial design measures this more meaningful quantity by sampling diverse configurations, at the cost of more total rollouts. In practice, random sampling of configurations (rather than exhaustive enumeration) keeps the cost tractable: 387 scenarios $\times$ $R{=}20$ rollouts $= 7{,}740$ rollouts per model, parallelizable across containers to a few hours of wall time.

\section{Additional simulation results}
\label{app:simulations}

\subsection{Full coverage tables}

Table~\ref{tab:sim_full} reports the complete simulation results across all combinations of rollout counts and resampling strategies, including the heterogeneous within-app condition.

\begin{table}[h]
  \caption{Extended simulation results ($S{=}10$ scenarios/app, $3{\times}3{\times}3$ configurations, 95\% CI). Under heterogeneous within-app noise ($\sigma_{\text{scen}} = 0.08$, $\sigma_{\text{config}} = 0.03$), rollout-only coverage collapses to 60\%.}
  \label{tab:sim_full}
  \centering
  \small
  \begin{tabular}{@{}llcc@{}}
    \toprule
    Condition & Resampling & Coverage & Width \\
    \midrule
    \multirow{4}{*}{Homogeneous, $R{=}3$}
      & Scen + Config + Roll & 1.000 & 0.061 \\
      & Scen + Config        & 1.000 & 0.046 \\
      & Scen + Roll          & 1.000 & 0.024 \\
      & Roll only            & 0.860 & 0.013 \\
    \midrule
    \multirow{2}{*}{Heterogeneous, $R{=}3$}
      & Scen + Config + Roll & 1.000 & 0.065 \\
      & Roll only            & \textbf{0.600} & 0.013 \\
    \bottomrule
  \end{tabular}
\end{table}

\subsection{Sensitivity to number of bootstrap replicates}

We test $B \in \{100, 300, 500, 1000, 2000\}$ bootstrap replicates for the full hierarchical bootstrap with $R = 3$. Coverage remains at 100\% for all values of $B \geq 100$. CI width stabilizes at $B = 300$ (change $< 0.001$ for larger $B$), justifying our default choice of $B = 300$ for simulations and $B = 1{,}000$ for final evaluations.

\section{Detailed \textsc{DigiWorld} statistics}
\label{app:digiworld_stats}

\begin{table}[h]
  \caption{The \textsc{DigiWorld} application suite. 15 apps across 7 domains, each with ${\sim}10$ data profiles, ${\sim}10$ themes, and ${\sim}10$ UI states.}
  \label{tab:apps}
  \centering
  \small
  \begin{tabular}{llll}
    \toprule
    Domain & Applications & Task types & Scenarios \\
    \midrule
    Communication & Email, Messaging & Send, search, organize & 39 \\
    Financial & Payment, Banking & Transfer, check balance, manage accounts & 64 \\
    Shopping & E-Commerce, QwikShop, Auction & Browse, purchase, bid, review & 115 \\
    Entertainment & Music, Video & Play, search, manage playlists & 42 \\
    On-Demand & Eats, Ryde & Order food, book rides & 34 \\
    Travel & Flight Booking, Transit & Search flights, plan routes & 49 \\
    Smart Living & Smart Home, Parking & Control devices, find parking & 44 \\
    \midrule
    \multicolumn{3}{l}{\textbf{Total (original + composed)}} & \textbf{387} \\
    \bottomrule
  \end{tabular}
\end{table}

\begin{table}[h]
  \caption{Per-app scenario counts and task type distribution in \textsc{DigiWorld}. Action tasks require the agent to modify application state; information retrieval (IR) tasks require the agent to extract and report information.}
  \label{tab:scenario_stats}
  \centering
  \small
  \begin{tabular}{lccc}
    \toprule
    Application & Action & IR & Total \\
    \midrule
    Email & 12 & 7 & 19 \\
    Messaging & 19 & 1 & 20 \\
    Payment & 16 & 11 & 27 \\
    Banking & 19 & 18 & 37 \\
    E-Commerce & 18 & 40 & 58 \\
    QwikShop & 22 & 2 & 24 \\
    Auction & 27 & 6 & 33 \\
    Music & 14 & 0 & 14 \\
    Video & 12 & 16 & 28 \\
    Eats & 9 & 5 & 14 \\
    Ryde & 11 & 9 & 20 \\
    Flight Booking & 7 & 10 & 17 \\
    Transit & 13 & 19 & 32 \\
    Smart Home & 21 & 6 & 27 \\
    Parking & 17 & 0 & 17 \\
    \midrule
    \textbf{Total} & \textbf{237} & \textbf{150} & \textbf{387} \\
    \bottomrule
  \end{tabular}
\end{table}

\section{Sources of variability in \textsc{DigiWorld}}
\label{app:variability_axes}

\begin{table}[h]
  \caption{The five sources of variability in \textsc{DigiWorld}. The four environmental axes are independently toggleable; rollout stochasticity is always present. When all environmental axes are disabled, the benchmark is deterministic and the only uncertainty comes from the agent's sampling.}
  \label{tab:variability_axes}
  \centering
  \footnotesize
  \resizebox{\textwidth}{!}{%
  \begin{tabular}{llll}
    \toprule
    Source & Type & What varies & Example \\
    \midrule
    Instance params & Environmental & Task parameters (amounts, recipients, items) & ``Send \$50'' vs.\ ``Send \$200'' \\
    Data profile    & Environmental & Database content (accounts, messages, history) & 3 contacts vs.\ 50 contacts \\
    Theme           & Environmental & Visual styling (colors, fonts, layout density) & Light mode vs.\ dark mode \\
    UI state        & Environmental & Starting screen and navigation state & Home screen vs.\ settings page \\
    \midrule
    Rollouts        & Algorithmic   & Agent's stochastic sampling (seed, temperature) & Same screen, different actions \\
    \bottomrule
  \end{tabular}%
  }
\end{table}

\section{Scenario lifecycle details}
\label{app:scenario_lifecycle}

A \emph{scenario} in \textsc{DigiWorld} defines a parameterized task template (e.g., ``Send \$\texttt{<amount>} to \texttt{<recipient>}'') together with its verification logic. Each scenario is instantiated into multiple \emph{instances} by binding concrete values to the template parameters (e.g., ``Send \$50 to Alice,'' ``Send \$200 to Bob''). Instance randomization (Table~\ref{tab:variability_axes}) controls whether a fixed default instance or a randomly selected instance is used at evaluation time. Each instance is compatible with a subset of the available data profiles, determined by automated feasibility checking (Section~\ref{sec:integrity}). A \emph{configuration} is then the combination of an instance with a specific data profile, theme, and UI state drawn from whichever environmental axes are enabled. The scenario lifecycle proceeds in three phases:

\textbf{1. Reset.}~~ The framework selects a compatible profile, pushes the corresponding database and UI state files to the emulator via ADB, and injects instance-specific mockdata through a context extraction pipeline that resolves template placeholders against the profile's database (Section~\ref{sec:integrity}).

\textbf{2. Execution.}~~ The agent interacts with the application through the standard Android interface. \textsc{DigiWorld} is agent-agnostic: any system that can send touch/type events to an Android emulator and receive screenshots can be evaluated.

\textbf{3. Verification.}~~ The framework captures the final application state and compares it to the initial state using parameterized SQL queries. For \emph{action} tasks (e.g., sending an email), verification checks that the expected records were created or modified. For \emph{information retrieval} tasks (e.g., ``What is the account balance?''), the framework extracts the ground-truth answer from the database and compares it to the agent's response using flexible matchers that handle numeric formats, date representations, durations, and Unicode normalization.

\section{Variability decomposition: methodology and per-app details}
\label{app:variability}

For a given axis (e.g., theme), we construct matched pairs of configurations that differ only along that axis. Let $c_1 = (i, d, t_1, u)$ and $c_2 = (i, d, t_2, u)$ be two configurations with the same instance $i$, data profile $d$, and UI state $u$ but different themes $t_1$ and $t_2$. The performance difference $\Delta = \hat{p}_{c_1} - \hat{p}_{c_2}$ isolates the effect of the theme change. The same methodology applies symmetrically to all four environmental axes. We summarize the instability along each axis using the \textbf{Mean Absolute Deviation (MAD)}:
\begin{equation}
\text{MAD} = \frac{1}{|\mathcal{P}|}\sum_{(c_1,c_2) \in \mathcal{P}} |\Delta_{c_1,c_2}|
\label{eq:mad}
\end{equation}
which measures the average magnitude of performance shifts when a single axis changes, regardless of direction.

\begin{figure}[h]
  \centering
  \includegraphics[width=\textwidth]{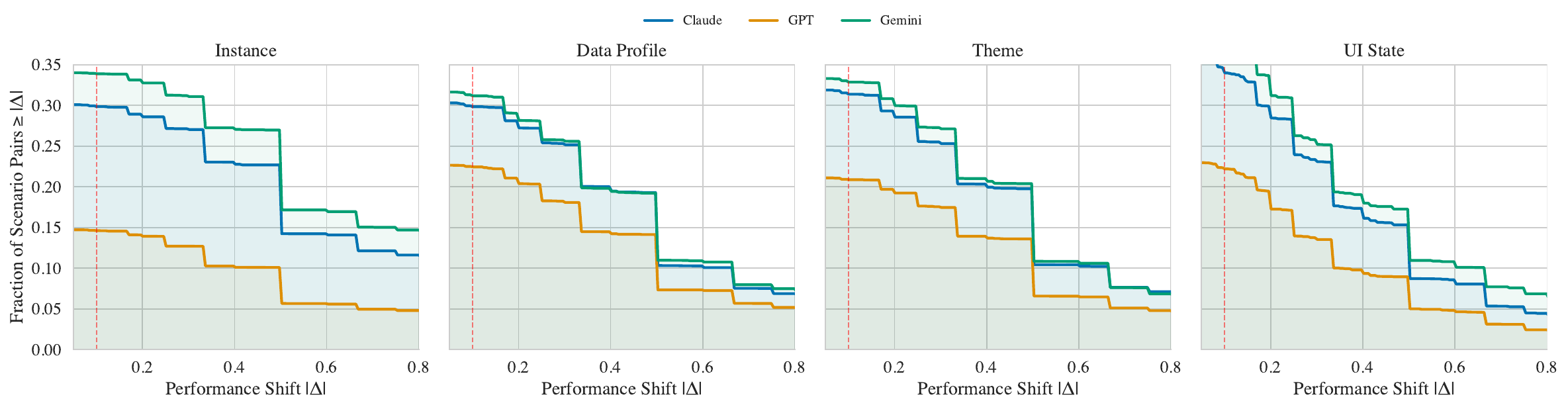}
  \caption{Exceedance curves for scenario-level performance shifts, per model and environmental axis. For each axis, we compute per-scenario success rates under each axis value (averaging across all other axes), take pairwise $|\Delta|$, and plot the fraction of pairs exceeding each threshold. Higher curves indicate greater sensitivity. The dashed red line marks $|\Delta|{=}0.10$.}
  \label{fig:variability}
\end{figure}

\begin{table}[h]
  \caption{Per-app MAD (Mean Absolute Deviation) for Claude across each variability axis.}
  \label{tab:per_app_variability}
  \centering
  \small
  \begin{tabular}{lcccc}
    \toprule
    Application & Instance MAD & Data MAD & Theme MAD & UI State MAD \\
    \midrule
    Email       & .200 & .190 & .138 & .170 \\
    Messaging   & .165 & .136 & .133 & .087 \\
    Payment     & .146 & .123 & .145 & .123 \\
    Banking     & .156 & .132 & .116 & .137 \\
    E-Commerce  & .233 & .222 & .224 & .211 \\
    QwikShop    & .112 & .073 & .051 & .061 \\
    Auction     & .083 & .133 & .144 & .129 \\
    Music       & .249 & .274 & .203 & .229 \\
    Video       & .346 & .300 & .308 & .244 \\
    Eats        & .109 & .104 & .090 & .065 \\
    Ryde        & .140 & .157 & .118 & .107 \\
    Flight      & .162 & .170 & .156 & .081 \\
    Transit     & .290 & .239 & .249 & .246 \\
    Smart Home  & .247 & .204 & .223 & .181 \\
    Parking     & .154 & .110 & .114 & .084 \\
    \midrule
    \textbf{Mean} & \textbf{.186} & \textbf{.171} & \textbf{.161} & \textbf{.144} \\
    \bottomrule
  \end{tabular}
\end{table}

\section{Per-app performance}
\label{app:per_app}

\begin{figure}[h]
  \centering
  \includegraphics[width=\textwidth]{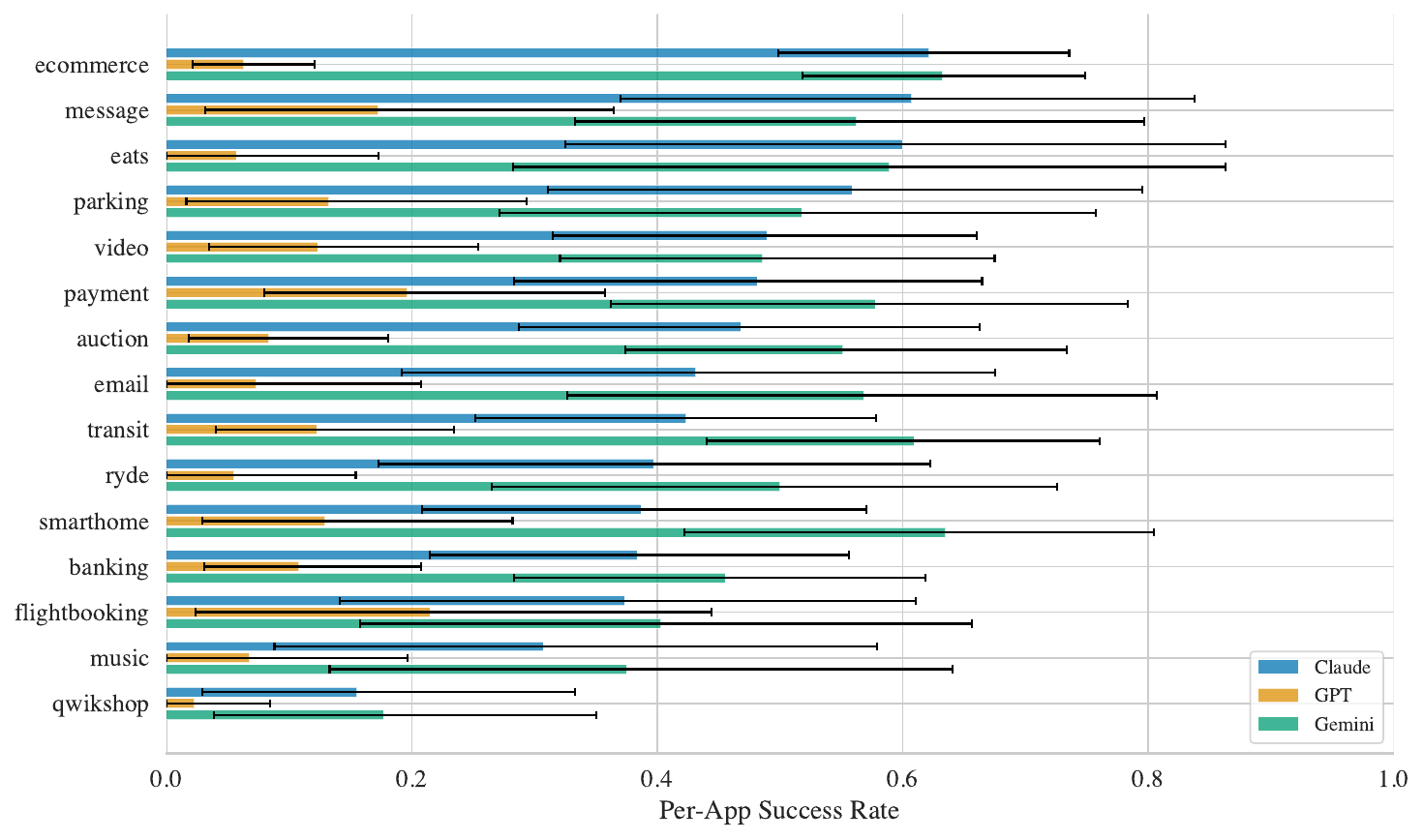}
  \caption{Per-app success rates for all three models, with 95\% hierarchical bootstrap CIs. Apps are ordered by Claude's performance.}
  \label{fig:per_app}
\end{figure}

\section{Additional model results}
\label{app:additional_models}

We additionally evaluate two non-frontier model variants---Claude Sonnet 4.6 and Gemini 3 Flash---to provide a broader picture of performance across model tiers. Figure~\ref{fig:results_appendix} shows suite-level scores and performance profiles for all five models.

\begin{figure}[h]
  \centering
  \includegraphics[width=\textwidth]{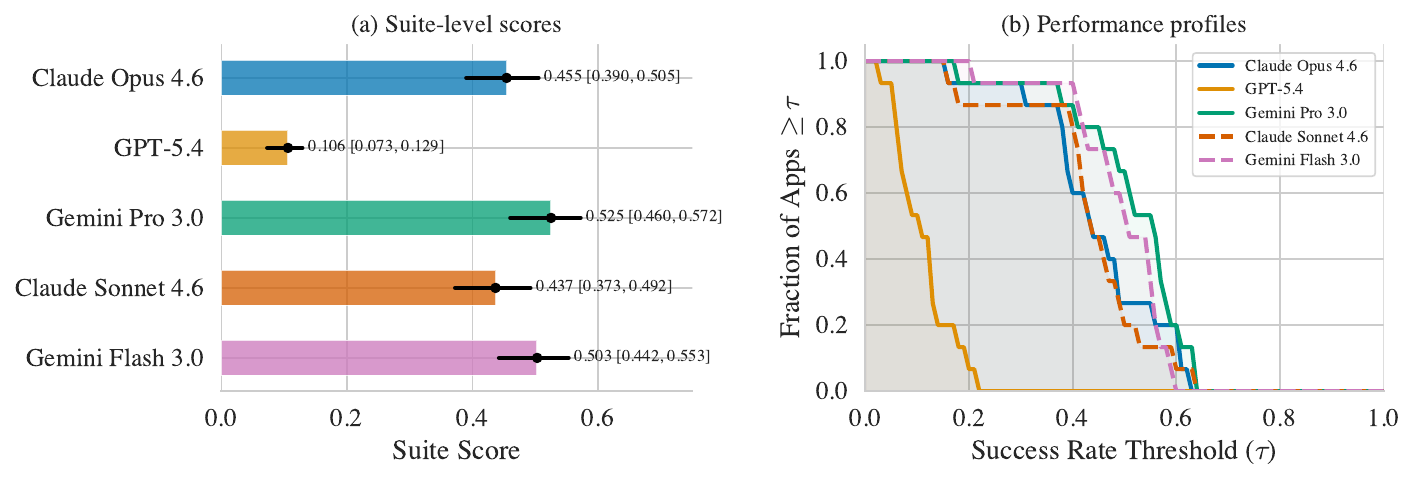}
  \caption{Suite-level scores and performance profiles for all five evaluated models. Solid lines in (b) denote the three frontier models from the main evaluation (Figure~\ref{fig:results}); dashed lines denote the additional non-frontier variants (Claude Sonnet 4.6 and Gemini 3 Flash).}
  \label{fig:results_appendix}
\end{figure}

\section{Integrity checking implementation details}
\label{app:integrity}

This appendix provides the implementation details for the three integrity mechanisms described in Section~\ref{sec:integrity}.

\subsection{Context extraction and grounded data injection}

Mockdata files use \texttt{\{\{placeholder\}\}} tokens that are resolved against the profile's SQLite database at reset time. Placeholders fall into three categories:

\begin{itemize}[leftmargin=*,itemsep=1pt]
  \item \textbf{Identity templates} extract values matching the profile's user identity: \texttt{\{\{current\_user\_email\}\}}, \texttt{\{\{current\_user\_id\}\}}, \texttt{\{\{current\_user\_name\}\}}.
  \item \textbf{Relational templates} extract values from existing records to ensure referential consistency: \texttt{\{\{first\_room\_id\}\}} resolves to an actual room ID in the smart home database, \texttt{\{\{context\_sender:work\}\}} resolves to a contact classified as ``work'' in the email app.
  \item \textbf{Positioning templates} compute values \emph{relative to} existing data: \texttt{\{\{middle\_transaction\_time\}\}} queries the transactions table, sorts by timestamp, and returns a time midway through the list. A positioning service implements strategies (\texttt{Beginning}, \texttt{Middle}, \texttt{End}) that place injected records at natural positions within the existing data distribution.
\end{itemize}

Each application defines its own template resolver that extends a base resolver with app-specific templates (e.g., \texttt{\{\{safe\_amount\}\}} in the payment app resolves to an amount guaranteed to be below the wallet balance). The resolver walks the mockdata structure recursively, resolving every placeholder before the data is merged into the profile's database. The same mechanism resolves the \emph{scenario context}, i.e., structured metadata (e.g., user PIN, account number) passed to the verification system for ground-truth comparison.

This design ensures that the \emph{same} scenario instance produces \emph{different} concrete tasks on different profiles: ``find the most recent transaction'' yields different amounts, dates, and descriptions depending on which profile's transaction history is loaded.

\subsection{Constraint-based feasibility checking}

The feasibility pipeline evaluates each profile--instance pair offline and produces a per-instance list of compatible profiles. Four constraint types cover the space of preconditions:

\begin{itemize}[leftmargin=*,itemsep=1pt]
  \item \texttt{EntityExists}: the profile's database must contain at least $n$ rows in a table, optionally filtered by user or column value (e.g., ``at least 1 account for the current user'').
  \item \texttt{DataVolume}: a stricter variant requiring at least $n$ rows matching a compound filter (e.g., ``at least 2 completed transactions'').
  \item \texttt{Balance}: a numeric field must meet a threshold that can depend on the instance parameters (e.g., ``wallet balance $\geq$ the transfer amount \texttt{<amount>}'').
  \item \texttt{MaxCount}: the profile must have \emph{at most} $n$ rows in a table, for scenarios that add records to capacity-limited collections (e.g., ``fewer than 3 saved addresses'').
\end{itemize}

Beyond explicit constraints, the framework \textbf{auto-derives} additional constraints from the mockdata templates. If mockdata references \texttt{\{\{first\_room\_id\}\}}, the system automatically adds an \texttt{EntityExists} constraint requiring at least one room in the database. This closes the gap between what the templates \emph{assume} exists and what the constraints \emph{check} exists.

\subsection{Triviality filtering}

The framework executes the full verification logic against the \emph{initial} state of each compatible profile, before the agent acts. Any profile where the task is already completed by coincidence (e.g., the ``correct'' database record happens to exist in the starting state) is excluded. This prevents inflated success rates from configurations where no agent action was needed.

\section{Detailed PRISM evaluation of existing benchmarks}
\label{app:prism_details}

This appendix provides a detailed assessment of each benchmark evaluated in Table~\ref{tab:prism_eval}. We focus exclusively on environments where the agent operates from visual inputs (screenshots), excluding text-only DOM-based benchmarks. To ensure consistent and transparent ratings, we apply the following rubric for each PRISM criterion:

\begin{description}[leftmargin=1.5em,itemsep=3pt]
  \item[Privileged (P):] \cmark{} = all tasks verified via programmatic state inspection; {\raise.17ex\hbox{$\scriptstyle\sim$}} = mixed (some tasks use programmatic checks, others rely on LLM-as-judge); \xmark{} = no programmatic verification.
  \item[Realistic (R):] \cmark{} = uses real-world or high-fidelity clone applications with production-grade complexity; {\raise.17ex\hbox{$\scriptstyle\sim$}} = purpose-built applications that capture realistic interaction patterns and visual complexity but are not clones of real products; \xmark{} = toy or heavily simplified interfaces that lack real-world visual or functional complexity.
  \item[Integrity (I):] \cmark{} = automated offline pipeline verifies every configuration is feasible, coherent, and non-trivial; {\raise.17ex\hbox{$\scriptstyle\sim$}} = includes procedural initialization logic that sets up preconditions (e.g., clearing databases, injecting entities) but does not systematically verify feasibility, coherence, or non-triviality across all configurations; \xmark{} = no automated safeguards for configuration validity.
  \item[Sandboxed (S):] \cmark{} = fully self-contained with deterministic reset; {\raise.17ex\hbox{$\scriptstyle\sim$}} = mostly isolated but with partial external dependencies; \xmark{} = relies on live services or lacks deterministic reset.
  \item[Multifactorial (M):] \cmark{} = systematically varies $\geq$2 independent environmental axes (e.g., data content and visual appearance); {\raise.17ex\hbox{$\scriptstyle\sim$}} = single-axis parameterization (e.g., task parameters vary but UI and data context remain static); \xmark{} = completely static task set with no parameterization.
\end{description}

A critical, often overlooked dimension of Integrity is \emph{staleness}: if an environment is live or its data is updated, tasks that reference specific entities, prices, or UI layouts may become impossible to complete. An integrity-checked benchmark must automatically verify that its tasks remain feasible and coherent against the current environment state.

\subsection{AndroidWorld}

Developed by Google DeepMind, AndroidWorld is a comprehensive environment for evaluating agents on Android devices \citep{rawles2024}.

\begin{itemize}[leftmargin=*,itemsep=1pt]
  \item \textbf{Privileged (\cmark):} Task success is verified programmatically by inspecting the device's internal system state via the Android Debug Bridge (ADB) and UI element validation.
  \item \textbf{Realistic (\cmark):} Uses 20 real-world, unmodified Android applications, including OS-level apps and open-source apps from F-Droid.
  \item \textbf{Integrity-checked ({\raise.17ex\hbox{$\scriptstyle\sim$}}):} Tasks include dedicated initialization logic to set up preconditions (e.g., clearing databases, injecting required entities), providing procedural safeguards that reduce the chance of infeasible configurations. However, there is no automated offline pipeline that systematically verifies feasibility, coherence, and non-triviality across all randomly generated configurations before evaluation.
  \item \textbf{Sandboxed (\cmark):} Runs on a self-contained Pixel~6 emulator with a fixed Android~13 image, allowing deterministic resets of device state and time.
  \item \textbf{Multifactorial ({\raise.17ex\hbox{$\scriptstyle\sim$}}):} Dynamically constructs parameterized tasks by randomly generating entity values (names, phone numbers, dates) within natural language instructions. However, this is single-axis parameterization (task instance); it does not systematically vary independent axes like UI themes or initial app data states.
\end{itemize}

\subsection{OSWorld}

OSWorld tests multimodal agents on open-ended computer tasks across standard operating systems \citep{xie2024osworld}.

\begin{itemize}[leftmargin=*,itemsep=1pt]
  \item \textbf{Privileged (\cmark):} Verification relies on custom execution-based evaluation scripts that programmatically check the internal system state (files, app state).
  \item \textbf{Realistic (\cmark):} Employs real, unmodified applications on standard operating systems (Ubuntu, Windows, macOS).
  \item \textbf{Integrity-checked ({\raise.17ex\hbox{$\scriptstyle\sim$}}):} Each task includes an initial state setup configuration (e.g., pre-configured VM snapshots, file system state) that establishes preconditions, reducing the chance of infeasible starting states. However, there is no automated offline pipeline that systematically verifies feasibility, coherence, or non-triviality of the 369 task configurations.
  \item \textbf{Sandboxed (\cmark):} Utilizes virtualization and containerization (e.g., Docker) for deterministic resetting and reproducibility without relying on live internet access.
  \item \textbf{Multifactorial (\xmark):} Features a completely static test set of 369 fixed tasks. Each task has a fixed initial VM snapshot, a fixed goal, and a fixed evaluation script, with no parameterization or systematic variation.
\end{itemize}

\subsection{VisualWebArena}

An extension of the text-based WebArena, VisualWebArena introduces visually grounded tasks \citep{koh2024visualwebarena}.

\begin{itemize}[leftmargin=*,itemsep=1pt]
  \item \textbf{Privileged (\cmark):} Executes JavaScript within the browser to inspect the DOM and internal application state for programmatic verification.
  \item \textbf{Realistic (\cmark):} The underlying web applications (Classifieds, Shopping, Reddit) are functional clones of real platforms that faithfully reproduce the structural, visual, and interaction complexity of their production counterparts, populated with real or realistic data.
  \item \textbf{Integrity-checked (\xmark):} Tasks are generated from curated templates without automated offline checks for feasibility, non-triviality, or staleness during generation.
  \item \textbf{Sandboxed (\cmark):} Fully sandboxed using Docker containers for each web application, enabling deterministic resets.
  \item \textbf{Multifactorial ({\raise.17ex\hbox{$\scriptstyle\sim$}}):} Uses parameterized task templates to generate unique tasks, but variability is restricted to swapping parameters within the same static UI and data environment.
\end{itemize}

\subsection{WorkArena}

WorkArena evaluates agents on enterprise software platforms \citep{drouin2024workarenacapablewebagents}.

\begin{itemize}[leftmargin=*,itemsep=1pt]
  \item \textbf{Privileged (\cmark):} Completion is verified programmatically using validation functions that check the internal state of the ServiceNow application.
  \item \textbf{Realistic (\cmark):} Uses the real, unmodified ServiceNow platform.
  \item \textbf{Integrity-checked (\xmark):} Lacks a comprehensive, automated integrity check for all task configurations, relying instead on manual curation and LLM generation without strict offline validation against the live platform state.
  \item \textbf{Sandboxed (\xmark):} The environment is remote-hosted and requires access to live ServiceNow instances, lacking deterministic reset or snapshotting capabilities.
  \item \textbf{Multifactorial ({\raise.17ex\hbox{$\scriptstyle\sim$}}):} Features thousands of parameterized task instances, but variability is limited to single-axis parameterization (e.g., different questions or forms) rather than independent axes of UI or state variation.
\end{itemize}

\subsection{REAL-bench}

REAL-bench tests agents on deterministic simulations of real-world websites \citep{garg2025real}.

\begin{itemize}[leftmargin=*,itemsep=1pt]
  \item \textbf{Privileged ({\raise.17ex\hbox{$\scriptstyle\sim$}}):} Uses a hybrid approach: action-based tasks are verified via programmatic state-diff mechanisms, but information-retrieval tasks rely entirely on an LLM-as-judge.
  \item \textbf{Realistic (\cmark):} Uses high-fidelity deterministic replicas of popular websites (e.g., Airbnb, Amazon, Gmail) that faithfully mirror the structural and functional complexity of the production sites.
  \item \textbf{Integrity-checked (\xmark):} Tasks are manually curated without automated verification for feasibility or coherence.
  \item \textbf{Sandboxed (\cmark):} A fully self-contained, controlled sandbox using deterministic simulations.
  \item \textbf{Multifactorial (\xmark):} Focuses on standardized simulations without systematic variation across factors like UI themes or initial data states; tasks are fixed single-instance scenarios.
\end{itemize}

\subsection{DigiData-Bench}

DigiData-Bench evaluates general-purpose mobile control agents on complex real-world tasks \citep{sun2025}.

\begin{itemize}[leftmargin=*,itemsep=1pt]
  \item \textbf{Privileged (\xmark):} Task success is evaluated either by human operators or via an LLM-as-judge analyzing screenshots and UI trees. There is no programmatic verification of the internal system state.
  \item \textbf{Realistic (\cmark):} Uses real-world Android applications on actual Android devices.
  \item \textbf{Integrity-checked (\xmark):} Uses manual ``state initialization instructions'' for human operators to set up prerequisites, lacking an automated offline verification pipeline.
  \item \textbf{Sandboxed (\xmark):} Runs on live devices interacting with applications in the wild, requiring human operators to monitor and intervene to prevent unsafe actions.
  \item \textbf{Multifactorial (\xmark):} Consists of a fixed set of 309 goals without parameterized variation or independent axes of complexity.
\end{itemize}

\subsection{WebVoyager}

WebVoyager evaluates end-to-end web agents on real-world websites \citep{he2024webvoyager}.

\begin{itemize}[leftmargin=*,itemsep=1pt]
  \item \textbf{Privileged (\xmark):} Relies entirely on GPT-4V as an LLM-judge to evaluate task success based on screenshots, with no programmatic state verification.
  \item \textbf{Realistic (\cmark):} Agents interact directly with live, real-world production websites (e.g., Apple, Booking.com, Google Flights).
  \item \textbf{Integrity-checked (\xmark):} Tasks are manually compiled without automated feasibility or coherence checks against the live, constantly drifting websites.
  \item \textbf{Sandboxed (\xmark):} Operates on the live internet, making the environment uncontrolled and susceptible to UI updates, A/B tests, and network latency.
  \item \textbf{Multifactorial (\xmark):} Features a static set of 643 tasks without systematic variation.
\end{itemize}

\subsection{MobileWorld}

MobileWorld benchmarks autonomous mobile agents in agent-user interactive and MCP-augmented environments \citep{kong2025}.

\begin{itemize}[leftmargin=*,itemsep=1pt]
  \item \textbf{Privileged (\cmark):} Employs multiple deterministic verification methods, including textual answer verification, backend database queries (e.g., PostgreSQL for Mattermost), local storage inspection via ADB, and application callbacks, entirely avoiding LLM-as-judge for task success.
  \item \textbf{Realistic (\cmark):} Uses real Android applications, including a mix of stock apps and production-grade self-hosted open-source alternatives (Mattermost, Mastodon, Mall4Uni) that faithfully reproduce real-world functional complexity.
  \item \textbf{Integrity-checked ({\raise.17ex\hbox{$\scriptstyle\sim$}}):} Uses snapshot-based state management to ensure reproducible starting conditions, providing deterministic preconditions for each task. However, there is no automated offline pipeline to verify the feasibility, coherence, or non-triviality of the 201 manually created tasks.
  \item \textbf{Sandboxed (\cmark):} Fully containerized using Docker-in-Docker with rooted Android Virtual Devices and self-hosted backends, ensuring complete isolation from external dependencies and deterministic reset capabilities.
  \item \textbf{Multifactorial (\xmark):} Features a static set of 201 fixed tasks, each starting from a predetermined snapshot. It does not parameterize tasks or systematically vary independent axes like UI theme or initial data.
\end{itemize}

\subsection{OpenApps}

OpenApps focuses on measuring UI-agent reliability by simulating environment variations \citep{ullrich2025}.

\begin{itemize}[leftmargin=*,itemsep=1pt]
  \item \textbf{Privileged (\cmark):} Reward and verification are based on programmatic checks of the full underlying application state.
  \item \textbf{Realistic ({\raise.17ex\hbox{$\scriptstyle\sim$}}):} Uses purpose-built synthetic applications (calendar, messenger, maps) that capture realistic interaction patterns and visual complexity, but are not clones of real-world products and lack the full functional depth of production applications.
  \item \textbf{Integrity-checked (\xmark):} Tasks are generated based on templates without a separate, automated validation step for integrity checking.
  \item \textbf{Sandboxed (\cmark):} A self-contained Python-based environment running on a single CPU, allowing for scalable, deterministic resets.
  \item \textbf{Multifactorial (\cmark):} Explicitly designed to systematically vary appearance (e.g., dark theme, font size) and content independently to test agent generalization across thousands of app versions.
\end{itemize}

\subsection{MiniWoB++}

MiniWoB++ is a collection of simple web interaction environments \citep{liu2018}.

\begin{itemize}[leftmargin=*,itemsep=1pt]
  \item \textbf{Privileged (\cmark):} Controlled via Selenium WebDriver, providing full access to the DOM and page state for programmatic verification.
  \item \textbf{Realistic (\xmark):} Uses synthetic toy interfaces and heavily simplified clones of web applications with surrogate backends.
  \item \textbf{Integrity-checked (\xmark):} Tasks are manually created or templated without systematic integrity checks for each instance.
  \item \textbf{Sandboxed (\cmark):} Runs locally via Selenium without requiring internet access, ensuring deterministic resets.
  \item \textbf{Multifactorial ({\raise.17ex\hbox{$\scriptstyle\sim$}}):} Incorporates parameter randomization for task goals, but lacks systematic variation of visual appearance or data context as independent axes.
\end{itemize}

\section{Simulation calibration from evaluation data}
\label{app:sim_calibration}

Both simulation studies in Figure~\ref{fig:simulation} are calibrated from \textsc{DigiWorld} evaluation data to ensure realistic conditions.

\subsection{Base-level estimator calibration (Figure~\ref{fig:simulation}a)}

We extract per-configuration success rates from all three frontier model evaluations on \textsc{DigiWorld} (${\sim}19{,}500$ configurations total). The observed distribution is heavily bimodal: ${\sim}68\%$ of configurations have a 0\% success rate and ${\sim}32\%$ have a 100\% success rate, reflecting that most configurations are either very hard or very easy for current models at $R{=}3$ rollouts.

The observed rates are discretized (only 0, $\tfrac{1}{3}$, $\tfrac{2}{3}$, 1 are possible at $R{=}3$), so we recover continuous true probabilities via a Jeffreys posterior: for each observed $k$ successes in $R$ rollouts, we draw $p_{\text{true}} \sim \text{Beta}(k + \tfrac{1}{2}, R - k + \tfrac{1}{2})$. This standard noninformative posterior converts the degenerate point masses at 0 and 1 into realistic continuous distributions (e.g., a 0/3 observation yields $p_{\text{true}} \sim \text{Beta}(0.5, 3.5)$ with mean 0.125). We then simulate $R$ fresh rollouts from each $p_{\text{true}}$ and compute Wald and Wilson CIs.

\subsection{Hierarchical bootstrap calibration (Figure~\ref{fig:simulation}b)}

The simulation uses $A{=}15$ apps, $S{=}8$ scenarios per app, $3{\times}3{\times}3$ configurations per scenario, and $R{=}3$ rollouts. Three parameters are calibrated from evaluation data:

\begin{itemize}[leftmargin=*,itemsep=2pt]
  \item \textbf{App-level success probabilities.} We use the observed per-app success rates from a representative mid-range model (ranging from 0.16 to 0.62, mean 0.41), rather than the synthetic $\text{linspace}(0.05, 0.95)$ used in prior work.

  \item \textbf{Scenario noise ($\sigma_{\text{scen}}$).} We compute the standard deviation of per-scenario success rates within each app, averaged across apps and models. The raw observed value (${\sim}0.39$) is only slightly inflated by rollout sampling noise (the per-scenario rate averages over ${\sim}57$ binary outcomes), yielding a corrected $\sigma_{\text{scen}} \approx 0.25$. This large value reflects the genuine difficulty spread within apps: some scenarios (e.g., ``send an email'') are much easier than others (e.g., ``find and forward a specific thread'').

  \item \textbf{Configuration noise ($\sigma_{\text{config}}$).} The observed per-configuration standard deviation within scenarios (${\sim}0.19$) is almost entirely explained by rollout sampling noise at $R{=}3$ (expected $\sqrt{p(1{-}p)/R} \approx 0.28$), leaving a corrected $\sigma_{\text{config}} \approx 0.05$. This small value indicates that, at $R{=}3$, the dominant source of between-configuration variation is rollout stochasticity rather than true environmental effects, a finding consistent with the relatively modest MADs observed in Section~\ref{sec:decomposition}. Larger $R$ would resolve finer true environmental effects.
\end{itemize}

We run 200 independent experiments with $B{=}500$ bootstrap replicates for each of five methods (Wald CI, and four cumulative bootstrap variants).

\textbf{Note on robust aggregation.}~~ For benchmarks with a larger number of applications or where some applications may produce anomalous scores (e.g., due to degenerate tasks or extreme difficulty mismatches), replacing the plain mean with a trimmed mean $\text{TM}_\alpha$ can improve robustness. In \textsc{DigiWorld}'s current setting ($A{=}15$ curated apps on a common 0--1 scale), the plain mean and 10\% trimmed mean produce nearly identical results, so we use the simpler estimator. The hierarchical bootstrap framework supports any suite-level statistic, including trimmed means, medians, or interquartile means, as a drop-in replacement.


\section{Application gallery}
\label{app:app_gallery}

\textsc{DigiWorld} comprises 15 purpose-built Android applications spanning five domains: communication (Email, Messaging), finance (Payment, Banking), shopping (E-Commerce, QwikShop, Auction), entertainment (Music, Video), on-demand services (Eats, Ryde), travel (Flight Booking, Transit), and smart living (Smart Home, Parking). Each application is a fully functional React Native app with local SQLite storage, custom theming support, and deep-link--driven state management. Figure~\ref{fig:app_gallery} shows the default home screen of each application.

\begin{figure}[h]
  \centering
  \includegraphics[width=\textwidth]{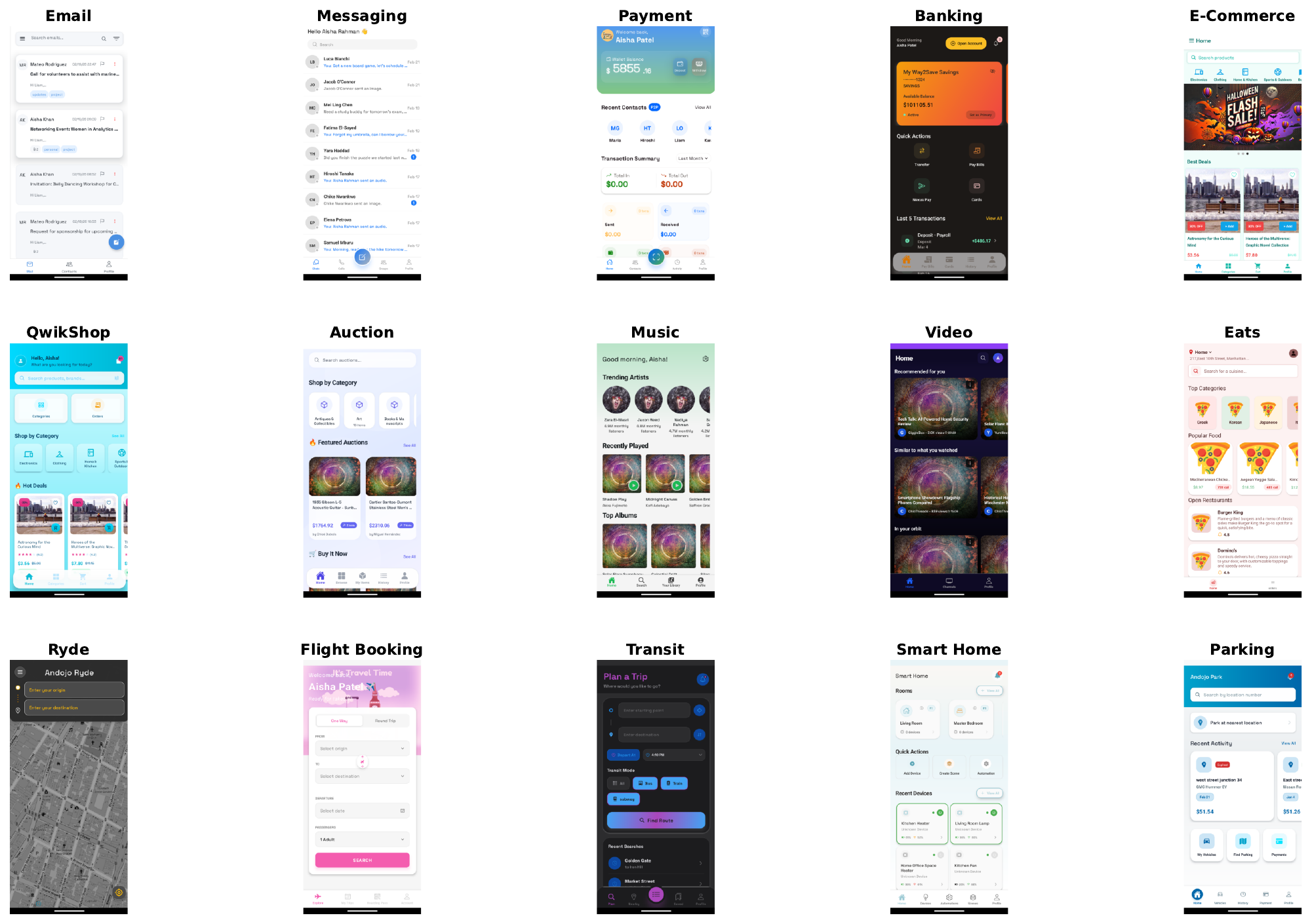}
  \caption{The 15 \textsc{DigiWorld} applications shown in their default configuration. Each app is a self-contained, sandboxed Android application with realistic UI and data.}
  \label{fig:app_gallery}
\end{figure}

\section{Visual impact of variability axes}
\label{app:variability_visuals}

Figure~\ref{fig:variability_axes} illustrates how each of the four environmental variability axes (Table~\ref{tab:variability_axes}) affects the visual appearance of four representative applications. For each app, we show the default configuration alongside variants produced by changing a single axis:

\begin{itemize}[leftmargin=*,itemsep=1pt]
  \item \textbf{Instance}: different task parameters are bound to the scenario template, injecting different concrete data into the app (e.g., different email subjects and senders, different wallet balances, different restaurant and category names).
  \item \textbf{Data Profile}: loads an entirely different database profile, changing all visible content (e.g., different users, transaction histories, menu items).
  \item \textbf{Theme}: switches the visual styling to a dark palette, changing background colors, text contrast, and accent tones.
  \item \textbf{UI State}: starts the app on a different screen (e.g., Contacts, Transactions, Orders, Cards), altering the navigation context.
\end{itemize}

These visual differences are representative of the environmental variation that agents must handle robustly. All four axes are independently composable: a single evaluation configuration combines one instance, one data profile, one theme, and one UI state, yielding over 3.2 million verified unique configurations across the benchmark.

\begin{figure}[h]
  \centering
  \includegraphics[width=\textwidth]{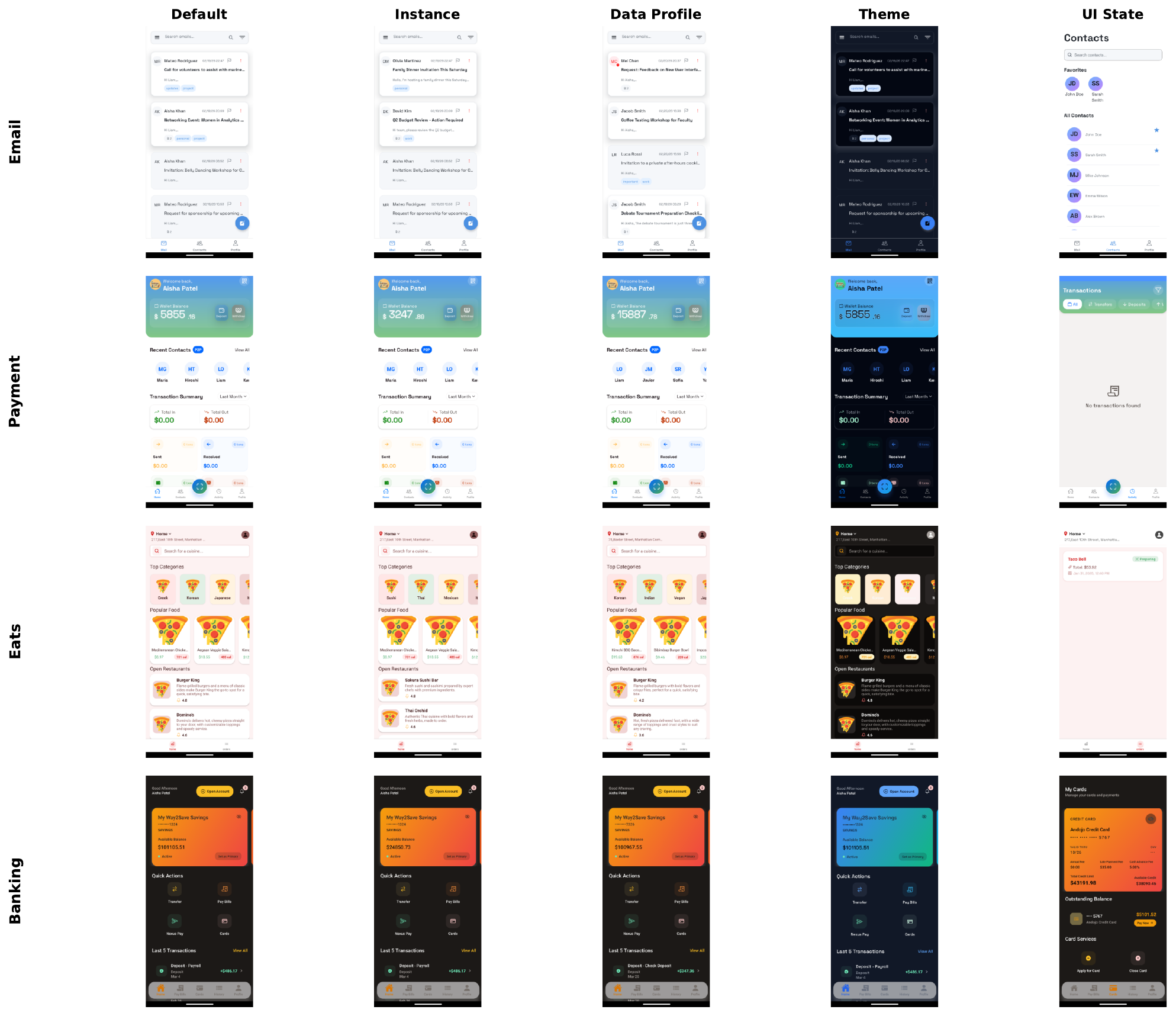}
  \caption{Visual impact of each environmental variability axis on four representative applications. Each column corresponds to one of the four environmental axes in Table~\ref{tab:variability_axes}; the leftmost column shows the default configuration for reference. Rows correspond to Email, Payment, Eats, and Banking.}
  \label{fig:variability_axes}
\end{figure}

\end{document}